\documentclass[12pt]{article}
\usepackage[dvipsnames,svgnames,x11names,hyperref]{xcolor}
\usepackage{amsmath}
\usepackage{amsthm}
\usepackage[T1]{fontenc}
\usepackage{comment}
\usepackage{fullpage}
\usepackage{authblk}
\usepackage{thmtools}

\usepackage[colorlinks]{hyperref}
\usepackage[normalem]{ulem}
\hypersetup{colorlinks={true},linkcolor={blue},citecolor=magenta}
\usepackage{cleveref}

\usepackage{amssymb,amsfonts,amsmath,amsthm,amscd,dsfont,mathrsfs}
\usepackage{bbm}
\usepackage{complexity}
\usepackage{tikz}
\usetikzlibrary{decorations.pathreplacing,backgrounds}

\usepackage{geometry}
\usepackage[width=.88\textwidth]{caption}

\usepackage{multirow}
\usepackage{mathpazo}
\usepackage{braket}
\usepackage{quantikz}

\usepackage[linesnumbered,ruled]{algorithm2e}

\usepackage[style=alphabetic,minalphanames=3,maxalphanames=4,maxnames=99,backref=true]{biblatex}

\usepackage{url}
\usepackage{bbm}
\usepackage{hyperref} 
\hypersetup{breaklinks=true}

  \theoremstyle{plain}
  \newtheorem{theorem}{Theorem}[section]
  \newtheorem{lemma}[theorem]{Lemma}
  
  \newtheorem{definition}[theorem]{Definition}
  \newtheorem{remark}[theorem]{Remark}

 \newtheorem*{theorem*}{Theorem}

\newcommand{\Tr}[1]{\mathrm{Tr}\left\{#1 \right\}}

\newcommand{\UES}{U_{\text{Enc}}^{S}}

\newcommand\algo{\mathcal}

\newcommand{\lsn}{\mathsf{LSN}}
\newcommand{\lpn}{\mathsf{LPN}}
\newcommand{\mslsn}{\mathsf{MSLSN}}

\newcommand{\plc}{\mathsf{PLC}}

\renewcommand{\vec}[1]{\mathbf{#1}}  
\newcommand{\bit}{\{0,1\}}
\newcommand{\Pn}{\mathcal{P}_n}
\newcommand{\reg}[1]{{\color{gray}#1}}
\newcommand{\ignore}[1]{}

\newcommand{\Cliff}{\mathsf{Cliff}}
\newcommand{\Stab}{\mathsf{Stab}}
\newcommand{\vCliff}{\mathsf{\text{\v{C}}liff}}
\newcommand{\Z}{\mathbb{Z}}
\newcommand{\Ber}{\mathsf{Ber}}

\newcommand{\kb}[1]{|#1\rangle\langle #1|} %

\newcommand{\negl}{\mathsf{negl}}
\renewcommand{\Mod}[1]{\ (\mathrm{mod}\ #1)}

\newcommand\calP{\mathcal{P}}

\renewcommand{\vec}[1]{\mathbf{#1}}
\newcommand\id{\mathbb{I}}
\newcommand{\linear}{\mathrm{L}}
\newcommand{\ot}{\otimes}

\title{The {\em Learning Stabilizers with Noise} problem \vspace{0.2cm}}

\author[1,2]{Alexander Poremba\footnote{\url{poremba@mit.edu}}}
\author[1]{Yihui Quek\footnote{\url{yquek@mit.edu}}}
\author[1]{Peter Shor\footnote{\url{shor@mit.edu}}}
\affil[1]{Department of Mathematics, Massachusetts Institute of Technology}
\affil[2]{Computer Science \& Artificial Intelligence Laboratory (CSAIL),\newline Massachusetts Institute of Technology}

\date{ }

\begin{document}

\maketitle

\abstract{Random classical codes have good error correcting properties, and yet they are notoriously hard to decode in practice. Despite many decades of extensive study, the fastest known algorithms still run in exponential time. The \emph{Learning Parity with Noise} ($\lpn$) problem, which can be seen as the task of decoding a random linear code in the presence of noise,
has thus emerged as a prominent hardness assumption with numerous applications in both cryptography and learning theory.

Is there a natural quantum analog of the $\lpn$ problem? In this work, we introduce the \emph{Learning Stabilizers with Noise} ($\mathsf{LSN}$) problem, the task of decoding a random stabilizer code in the presence of local depolarizing noise. We give both polynomial-time and exponential-time quantum algorithms for solving $\mathsf{LSN}$ in various depolarizing noise regimes, ranging from extremely low noise, to low constant noise rates, and even higher noise rates up to a threshold.
Next, we provide concrete evidence that $\lsn$ is hard. First, we show that $\mathsf{LSN}$ includes $\lpn$ as a special case, which suggests that it is at least as hard as its classical counterpart. Second, we prove worst-case to average-case reductions for variants of $\mathsf{LSN}$.
We then ask: what is the computational complexity of solving $\mathsf{LSN}$? Because the task features quantum inputs, its complexity cannot be characterized by traditional complexity classes.
Instead, we show that the $\mathsf{LSN}$ problem lies in a recently introduced (distributional and oracle) \emph{unitary synthesis class}. Finally, we identify several applications of our $\mathsf{LSN}$ assumption, ranging from the construction of quantum bit commitment schemes to the computational limitations of learning from quantum data. 
}

\newpage
\tableofcontents

\newpage

\section{Introduction}

Coding theory has offered many valuable insights into the theory of computation, ranging from 
structural insights in complexity theory ~\cite{PCP, NLTS}, to the design of cryptographic primitives~\cite{Sha79,Stern93,FS96,McEliece1978,Alekhnovich03} and even to lower bounds in the field of computational learning theory~\cite{BKW03, 10.1109/FOCS.2006.51}.
The existence of asymptotically good error correcting codes, in particular, is a major cornerstone in the field. Thanks to the probabilistic method, we know that a random linear code already attains the so-called Gilbert-Varshamov bound~\cite{Gilbert1952ACO, Varshamov} with high probability. This suggests that asymptotically good error correcting codes not only exist in theory, but are in fact also abundant. Despite their remarkable error correcting properties, random linear codes have been found to be notoriously hard to decode in practice, and the fastest known algorithms still run in exponential time~\cite{BKW03}.

\paragraph{Learning Parity with Noise.} The observation that better codes seem harder to decode is captured by the \emph{Learning Parity with Noise} ($\lpn$)
problem~\cite{BFKL93}. In a nutshell, this assumption says that it is computationally difficult to decode a random linear code under Bernoulli noise. In other words, given as input
$$
(\vec A \sim \Z_2^{n\times k},\vec A \cdot \vec x +  \vec e \Mod{2}) 
$$
it is hard to find the secret string $\vec x$
which is chosen uniformly at random in $\Z_2^{k}$, and where $\vec e \sim \Ber_p^{\otimes n}$ is a random Bernoulli error for some appropriate noise rate $p\in (0,1/2)$. Here, $\vec A \in \Z_2^{n\times k}$ serves as a random \emph{generator matrix} of a linear code, for $n = \poly(k)$.

In practice, $\lpn$ is believed to be hard for both classical and quantum algorithms running in time $\poly(k)$ in various noise regimes. For constant noise rates $p \in (0,1/2)$, the celebrated BKW algorithm~\cite{BKW03} solves $\lpn$ in both time and sample\footnote{Here, the number of samples refers to the parameter $n$---the number of \emph{noisy linear equations} on $\vec x$.}  complexity given by $O(2^{k/\log k})$. 
The conjectured hardness of $\lpn$ has found applications in both cryptography~\cite{10.5555/647097.717000,Alekhnovich03,10.1007/11535218_18,10.1007/978-3-642-03356-8_35,10.1007/978-3-642-34961-4_40,DBLP:conf/cans/DavidDN14,applebaum_et_al:LIPIcs.ITCS.2017.7,BLVW} and learning theory~\cite{BFKL93,10.1109/FOCS.2006.51}. The \emph{Learning with Errors} (LWE) problem~\cite{regev2009lattices}---a more structured variant of $\lpn$---has since become the basis of modern cryptography and has even led to highly advanced cryptographic primitives, such as fully homomorphic encryption~\cite{10.5555/1834954,6108154} and the classical verification of quantum computations~\cite{10.1137/20M1371828}. In the context of learning theory, it was shown that an efficient algorithm for $\lpn$ would allow us to learn important function classes, such as $2$-DNF formulas, juntas, and even more general functions with sparse Fourier spectrum~\cite{10.1109/FOCS.2006.51}.

Because the LPN problem is so prevalent in many areas of computer science, a significant effort has been devoted to finding evidence of its hardness. One of these pieces of evidence is a {\em worst-to-average-case reduction} \cite{BLVW, 10.1007/978-3-030-84252-9_16}. Recall that the $\lpn$ problem is an \emph{average-case} problem: the comoutational task is to decode a {\em random} code, secret and error. Ref.~\cite{BLVW} studied a related worst-case problem---the \emph{nearest codeword problem} (NCP)---and showed that it reduces to $\lpn$. This reduction, later improved by~\cite{10.1007/978-3-030-84252-9_16}, showed that $\lpn$ is at least as hard as (a mildly hard variant of) NCP in the worst case. \cite{BLVW} also found the first non-trivial complexity upper bound on the hardness of the $\lpn$ problem; specifically, they showed that (a variant of) the $\lpn$ problem is contained in the complexity class $\mathsf{Search}\mathsf{BPP}^{\mathsf{SZK}}$, and is thus unlikely to be $\mathsf{NP}$-hard.

\paragraph{The hardness of decoding random stabilizer codes.} Just like random linear codes, random \emph{quantum} stabilizer codes\footnote{The stabilizer formalism was first developed by Gottesman~\cite{gottesman1997stabilizercodesquantumerror} and incorporates the majority of quantum error correcting codes we know today~\cite{ErrorCorrectionZoo}.} also possess remarkable error correcting properties \cite{Graeme_thesis,gottesman_book}. They are ubiquitous in quantum information science; for example, random stabilizer codes appear in the context of quantum authentication schemes and the verification of quantum computations~\cite{aharonov2017interactiveproofsquantumcomputations}, quantum cryptography~\cite{dulek2018quantumciphertextauthenticationkey}, the theory of quantum communication~\cite{Graeme_thesis,Wilde_2013}, and even in black-hole physics and quantum gravity~\cite{Hayden_Preskill_2007,yoshida2017efficientdecodinghaydenpreskillprotocol,Harlow_2013}. Characterizing the hardness and complexity of decoding random stabilizer codes is therefore not only important from the perspective of quantum error correction, but could also shed a new light on the computational limitations of many natural quantum information processing tasks.
In this work, we ask: 
\begin{center}
\emph{How hard is it to decode a random quantum stabilizer code?}
\end{center}
Surprisingly, this subject has seen little theoretical treatment. While prior work has shown that decoding quantum stabilizer codes is {\em worst-case} hard~\cite{HsiehLeGall11,IP15}---via reduction from a purely classical decoding problem---the average-case complexity of decoding stabilizer codes as an inherently \emph{quantum problem} was left as an open problem~\cite{IP15}.
We re-formulate this as a question that bears on all of the areas mentioned above:

\begin{center}
\emph{Can we find a natural quantum analog of the Learning Parity with Noise problem? In particular, what would its hardness imply for quantum information science as a whole?}
\end{center}

Given the success of constructing cryptographic primitives from the hardness of LPN in a classical world, could such a quantum analog of LPN allow us to directly construct cryptographic protocols in a quantum world? 
Finally---and perhaps, even more interestingly---this assumption may turn out to be even harder to break than its classical counterpart.


\section{Overview}

 
We now give an overview of our contributions in this work, summarized in Table \ref{tab:lpn-lsn-comparison}.

\subsection{Learning Stabilizers with Noise}

In this work, we introduce a natural quantum analog of $\lpn$---the {\em Learning Stabilizers with Noise} (LSN) problem. In studying the LSN problem, we thoroughly characterize the hardness and complexity of decoding random stabilizer codes in different noise regimes. Similar to the $\lpn$ problem, which has found numerous applications in both cryptography and learning theory, we believe that our LSN assumption has the potential to occupy a similar role in quantum information more broadly.
Before we introduce our $\lsn$ task formally, we first revisit $\lpn$ and draw a connection to quantum error correction.

\begin{table}[h]
\centering
\begin{tabular}{|p{2.6cm}|p{5cm}|p{6.5cm}|}
\hline
 & Learning Parity with Noise & Learning Stabilizers with Noise (this problem) \\
\hline
Worst-case hardness & $\checkmark$ 
$\mathsf{NP}$-complete~\cite{BMT78} as a decisional syndrome decoding task.
Variant of the {\em (Promise) Nearest Codeword Problem} ($\mathsf{NCP}$)~\cite{BLVW} & As a classical syndrome decoding task:~$\mathsf{NP}$-complete~\cite{HsiehLeGall11,KuoLu12} or $\#\mathsf{P}$-complete~\cite{IP15} depending on the decoding problem\\
\hline
Average-case hardness & \vspace{0.001mm}$\checkmark$~\cite{BLVW, 10.1007/978-3-030-84252-9_16} & \vspace{0.001mm}This paper (\Cref{sec:complexity})\\
\hline
Worst-to-average-case reductions & \vspace{0.001mm}$\checkmark$~\cite{BLVW, 10.1007/978-3-030-84252-9_16} & \vspace{0.001mm}This paper (\Cref{sec:worst-avg-reduction})\\
\hline
Algorithms for average-case problem & $\checkmark$ $2^{O(k/\log k)}$ time/sample complexity \cite{blum2000noisetolerantlearningparityproblem} in constant-noise regime. & \vspace{0.001mm}This paper (\Cref{sec:algorithms})\\
\hline
\end{tabular}
\caption{Comparison of $\lpn$ and LSN in terms of hardness and complexity}
\label{tab:lpn-lsn-comparison}
\end{table}

\paragraph{A quantum analog of $\lpn$?} Let $n,k \in \mathbb{N}$ be integers with $n=\poly(k)$, and let $p \in (0,1/2)$ be a parameter. Recall that an instance of the $\lpn$ problem\footnote{For a more formal definition, see \Cref{def:LPN}.} consists of a generator matrix for a random linear code, together with a noisy codeword for a uniformly random string; specifically, we consider samples of the form
$$
(\vec A \sim \Z_2^{n\times k},\vec A \cdot \vec x + \vec e \Mod{2})
$$
where $\vec A \in \Z_2^{n\times k}$ is a random generator matrix, where $\vec A \cdot \vec x + \vec e \Mod{2}$ is a noisy codeword which encodes uniformly random string $\vec x \sim \Z_2^k$, and
where $\vec e \sim \Ber_p^{\otimes n}$ is a random Bernoulli error. Without loss of generality\footnote{This happens with overwhelming probability for $\vec A \sim \Z_2^{n\times k}$ provided that $n \gg k$ (see \Cref{subsec:reduction-LPN}).}, we assume that the matrix $\vec A$ has full column-rank, i.e., the columns of $\vec A$ generate all of $\Z_2^k$.
We now make a simple observation; namely, we can interpret the $\lpn$ instance $\vec A \cdot \vec x + \vec e \Mod{2}$ as a particular noisy quantum codeword on $n$ qubits\footnote{Strictly speaking, we should think of $\ket{\vec A\cdot \vec x + \vec e \Mod{2}}$ as encoding the row vector $\vec x^\intercal \vec A + \vec e^\intercal \Mod{2}$.}, since 
\begin{align}
\ket{\vec A\cdot \vec x + \vec e \Mod{2}}  &= X^{\vec e} \ket{\vec A\cdot  \vec x \Mod{2}} \nonumber \\
&= X^{\vec e} U_{\vec A}\left(\ket{0^{n-k}} \otimes \ket{\vec x}\right) \label{eq:noisy-LPN-codeword}\, ,
\end{align}
where $X^{\vec e} = X^{e_1} \otimes \dots \otimes X^{e_n}$ is a product of low-weight Pauli-X operators and where the unitary operator $U_{\vec A}$ is defined to be the matrix multiplication operation
\begin{align}\label{eq:encoding-U-A}
U_{\vec A}: \ket{0^{n-k}} \otimes \ket{\vec x} \rightarrow \ket{\vec A \cdot\vec x \Mod{2}}.
\end{align}
Because $\vec A$ has full column-rank, $U_{\vec A}$ corresponds to a linear reversible circuit which can be described solely in terms of $\mathrm{CNOT}$ gates~\cite{10.5555/2011763.2011767}. Therefore, $U_{\vec A}$ is a Clifford operator and thus maps Pauli operators to Pauli operators via conjugation.

We may also observe that $U_{\vec A}$ is the encoding circuit for a stabilizer code. The stabilizer group associated with this code is precisely the group of $k$ commuting Pauli operators under which $U_{\vec A}(\ket{0^{n-k}} \otimes \ket{\vec x})$ remains invariant. These are easily seen to be the Pauli operators 
$$
U_{\vec A}Z_i U_{\vec A}^\dag, \quad \text{ for } i \in [n-k] \, ,
$$
where $Z_i$ denotes a Pauli operator which acts as a Pauli-$Z$ operator on the $i$-th qubit, and is equal to the identity everywhere else.
In other words, the Clifford encoding unitary $U_{\vec A}$ (derived from an instance of an $\lpn$ problem) gives rise to the quantum stabilizer code\footnote{See \Cref{sec:stabilizer} for additional background on stabilizer codes.} given by
$
S_{\vec A} = \langle U_{\vec A}Z_1 U_{\vec A}^\dag, \dots, U_{\vec A}Z_{n-k} U_{\vec A}^\dag \rangle$.
This shows that every instance of $\lpn$ can be mapped to an instance of decoding stabilizer codes. 

We now generalize this idea significantly, ultimately leading us to define the \emph{Learning Stabilizers with Noise} (LSN) problem---the natural quantum analog of the $\lpn$ problem:
\begin{itemize}
    \item (Random stabilizer code:) Note that the encoding Clifford unitary in \Cref{eq:encoding-U-A} generates a specific stabilizer code of the form $S_{\vec A} = \langle U_{\vec A}Z_1 U_{\vec A}^\dag, \dots, U_{\vec A}Z_{n-k} U_{\vec A}^\dag \rangle$. 
    We consider stabilizer subgroups of the Pauli group which are chosen uniformly at random from the set of all stabilizer subgroups with $n-k$ generators, denoted by $\Stab(n,k)$. In fact, as we later prove in \Cref{thm:random-Stab}, this is equivalent to choosing random stabilizer codes which are described by the set of generators $\langle CZ_1C^\dag, \dots,CZ_{n-k}C^\dag \rangle$, where $C \sim \Cliff_n$ is a random $n$-qubit Clifford operator.\footnote{While this is considered \emph{folklore}, we are not aware of a proof which has previously appeared in the literature. Therefore, we decided to rigorously prove this statement in \Cref{sec:stabilizer}.}

    \item (Local depolarizing noise:) Recall that the noisy codeword $X^{\vec e} U_{\vec A}(\ket{0^{n-k}} \otimes \ket{\vec x})$ in \Cref{eq:noisy-LPN-codeword} is only affected by low-weight bit-flip errors $X^{\vec e}$, where $\vec e \sim \Ber_p^{\otimes n}$ comes from a Bernoulli distribution. In quantum systems, however, noise may also come in the form of phase errors. This leads us to consider a quantum noise model in the form of local depolarizing noise $\algo D_p^{\otimes n}$. Similar to the Bernoulli distribution, local depolarizing noise also produces low-weight errors with high probability, which therefore naturally generalizes the classical noise model.
\end{itemize}
In other words, we consider the task of decoding a random quantum stabilizer code in the presence of local depolarizing noise. Because the codeword is a \emph{stabilizer state}, we call this the \emph{Learning Stabilizers with Noise} (LSN) problem---in analogy to the classical $\lpn$ problem.
We now give a formal definition of the problem.

\paragraph{Learning Stabilizers with Noise.}
The \emph{Learning Stabilizers with Noise} (LSN) problem (formally defined in \Cref{def:LSN}) is to find $x \in \bit^k$ given as input a sample
$$
\big(S\in \Stab(n,k), \,E \ket{\overline{\psi_x}}^S \big) \quad \text{ with } \quad \ket{\overline{\psi_x}}^S:= U_{\mathrm{Enc}}^S(\ket{0^{n-k}} \otimes \ket{x}) \, ,
$$
where $S\sim \Stab(n,k)$ is a uniformly random stabilizer subgroup, where $E \sim \mathcal{D}_p^{\otimes n}$ is a Pauli error from a local depolarizing channel, where $x \sim \bit^k$ is a random string, and where $U_{\mathrm{Enc}}^S$ is some canonical encoding circuit for the stabilizer code associated with $S$. As mentioned before, the encoding circuit is typically given in the form of a random $n$-qubit Clifford operator.

At first sight, it may not be clear why the LSN problem is even well-defined, since a unique solution to the quantum learning problem may not exist in certain parameter regimes. The intuition behind our argument for the existence of a unique solution is as follows. Suppose that $p\in(0,1/2)$ is a sufficiently small constant. Then, the quantum Gilbert-Varshamov bound (see \Cref{sec:stabilizer}) tells us that a random stabilizer code is non-degenerate and has distance at least $d=3np+1$ with overwhelming probability, in which case for any pair of codewords with $x,y \in \bit^k$, we have $$\bra{\overline{\psi_x}} E_a^\dag E_b \ket{\overline{\psi_y}}=0$$
by the \emph{Knill-Laflamme conditions}---provided that $E_a, E_b$ have weight at most $|E_a|, |E_b| \leq \frac{3}{2}np$. Fortunately, a simple Chernoff bound analysis reveals that this is the case with overwhelming probability for the local depolarizing channel $\algo D_p^{\otimes n}$. Therefore, Pauli errors which originate from a local depolarizing noise channel take orthogonal codewords to orthogonal codewords, and hence there must exist a measurement that perfectly distinguishes between them. This observation is also at the core of our algorithms  for the LSN problem, which we describe next.

\paragraph{Algorithms for Learning Stabilizers with Noise.}

Previously, we discussed why random stabilizer codes give rise to a single-shot decoding problem which exhibits unique solutions with high probability. This suggests that LSN can be solved information-theoretically. Can we also find efficient algorithms for solving the LSN problem? Not surprisingly, the answer depends on the specific noise regime of the error distribution (see \Cref{fig:LSN-hardness}). 

In \Cref{sec:algorithms}, we give
both polynomial-time and exponential-time quantum algorithms for solving the LSN problem in various noise-regimes:

\begin{figure}[t]
    \centering

\includegraphics[width=0.82\linewidth]{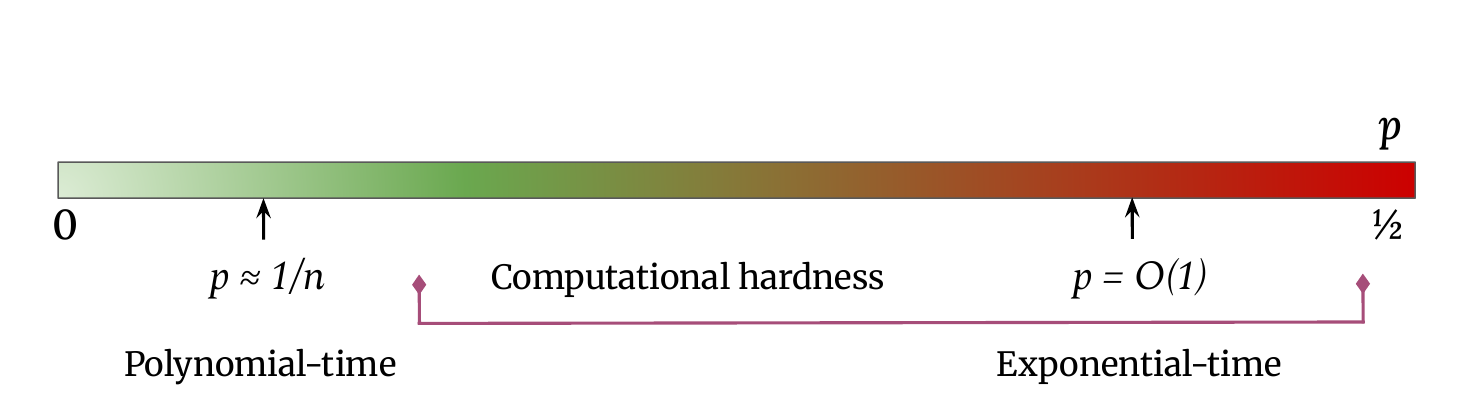}
    \caption{\textbf{Algorithms for the LSN problem.} In an extremely low noise regime with $p\approx 1/n$, the problem can be solved in polynomial-time. In a constant-noise regime with $p=O(1)$, the problem can be solved information-theoretically. Here, we give an exponential-time algorithm and conjecture that it is computationally hard. 
    }
    \label{fig:LSN-hardness}
\end{figure}

\begin{itemize}
    \item \textbf{Extremely low-noise} regime with parameter $p \approx \frac{1}{n}$. In this parameter regime, we show that a simple projection onto the stabilizer codespace (see Algorithm 1) suffices to solve the LSN problem in time $O(n^3)$ with constant success probability.

    \item \textbf{Low constant-noise} regime for a some small constant $p \in (0,1/2)$. In this regime, we show that the \emph{Pretty Good Measurement} (PGM)~\cite{barnum2000reversingquantumdynamicsnearoptimal,Montanaro_2007} (with only a single sample) succeeds with high probability. This is inextricably linked to the structure of our problem: although the distance between two arbitrary orthogonal states contracts tremendously---in fact, exponentially in system size (see, e.g. Proposition IV.7 in \cite{hirche2023quantumdifferentialprivacyinformation})---under a layer of  local depolarizing noise, a random stabilizer code encodes orthogonal states into \emph{orthogonal} subspaces that are ``protected" from such destructive contraction even under noise. This means that the information of the LSN secret is preserved under noise, and this is the intuition behind our approach in Algorithm 2.
    \item \textbf{Higher constant-noise} regime. In this regime, decoding is still possible at the cost of more samples. We derive the scaling of the sample complexity with noise, up to a certain noise threshold. 
\end{itemize}

\paragraph{Connection to syndrome decoding.}

Recall that an instance of the LSN problem by definition features \emph{quantum inputs}, which naturally led us to consider \emph{quantum} algorithms as a means of solving the problem. However, despite the fact that LSN is a quantum problem, it reduces to a natural \emph{classical} problem---\emph{stabilizer syndrome decoding}. One way of solving the problem is to just measure the stabilizer syndromes of the noisy quantum codeword and to tackle the classical decoding problem instead.
However, there are good reasons to believe that the classical (average-case) syndrome decoding problem is significantly harder. First, while LSN trivially reduces to stabilizer syndrome decoding, the reverse direction is highly unclear; in particular, it is not at all obvious why an algorithm for the LSN problem would also imply an algorithm for the classical stabilizer syndrome decoding problem.
Second, it is also evident from our algorithms in \Cref{sec:algorithms} that the LSN problem is fundamentally different from the classical syndrome decoding problem: as our quantum algorithm (see Algorithm $2$) suggests, it is possible to solve LSN in time which is merely exponential in $k$, whereas any naive algorithm for the syndrome decoding problem is likely to run in time which is exponential in $n$, where $n=\poly(k)$ is typically larger than $k$.

\paragraph{Worst-case to average-case reductions.}

Recall that the LSN decoding problem is stated as an \emph{average-case} problem, where the
success probability of an algorithm is measured on average over the random choice of stabilizer $S \in \Stab(n,k)$, secret $x \in \Z_2^k$ and error $E \sim \algo D_p^{\otimes n}$.
While the quantum Gilbert-Varshamov bound does in fact guarantee that an average-case instance of LSN can be solved information-theoretically, our results in \Cref{sec:algorithms} indicate that the problem becomes computationally intractable for large $k$---even in a low constant-noise regime. This raises the question of whether we can find concrete evidence for the average-case hardness of the LSN problem, beyond the fact that it subsumes the classical $\lpn$ problem.

Recently, a number of works showed that there is in fact evidence of \emph{worst-case} hardness for $\lpn$; specifically, by studying a related worst-case problem---the \emph{nearest codeword problem} (NCP)~\cite{BLVW,10.1007/978-3-030-84252-9_16}. Using the \emph{sample amplification} technique~\cite{10.1007/11538462_32}, Brakerski, Lyubashevsky, Vaikuntanathan and Wichs~\cite{BLVW} gave a worst-case to average-case reduction from NCP to $\lpn$. Here, an instance to the former problem consists of $
(\vec C \in \Z_2^{m \times k}, \, \vec t= \vec C \cdot \vec s + \vec w \Mod{2})$ for some $\vec s \in \Z_2^k$, with the promise that the generator matrix $\vec C$ is balanced\footnote{Roughly speaking, this means that the minimum and maximum distance of the linear code generated by $\vec C \in \Z_2^{n \times m}$ is neither too small nor too large.} and that the Hamming weight of the error $\vec w \in \Z_2^m$ is known. On a high level, the reduction in~\cite{BLVW} proceeds in two steps:
\begin{itemize}
    \item (Re-randomization of the secret) A random string $\vec u \sim \Z_2^k$ is chosen, and the worst-case instance $(\vec C,\vec t)$ gets mapped via an additive shift to $$(\vec C, \, \vec t + \vec C \cdot \vec u \Mod{2}) = (\vec C, \vec C \cdot (\vec s + \vec u) + \vec w \Mod{2}).$$
    Note that, whereas the initial secret $\vec s \in \Z_2^k$ was fixed, the new secret $\vec s + \vec u$ is now distributed according to the uniform distribution over $\Z_2^k$.

    \item (Re-randomization of the code and error) A random $\vec R \in \Z_2^{n \times m}$ is sampled from a \emph{smoothing distribution} $\algo R_w^{n \times m}$, and the previous sample gets mapped to
    $$
    (\vec R \cdot \vec C, \vec R \cdot(\vec C \cdot (\vec s + \vec u) + \vec w \Mod{2})) = (\vec R\cdot \vec C, \vec R \cdot \vec C \cdot (\vec s + \vec u) + \Vec R \cdot\vec w \Mod{2}))
    $$
\end{itemize}
\cite{BLVW} show that the resulting sample is statistically close to an (average-case) $\lpn$ sample---provided that the \emph{smoothing distribution} $\algo R_w^{n \times m}$ is chosen appropriately.

By taking a similar approach, we develop a worst-case to average-case reduction for the LSN problem. Here, the starting point is a worst-case stabilizer decoding instance $(S, E\ket{\overline{\psi}_x}^S)$ for some stabilizer $S\in \text{Stab}(n,k)$, Pauli error $E$ of bounded weight (at most $\log^c(n)$ for any $c>0$), and secret $x\in \{0,1\}^k$.
First, we observe that, in order to re-randomize the secret $x$, we need to act on the encoded data $\ket{\overline{\psi}_x}^S$ itself. Hence, it suffices to choose a random string $u \sim \bit^k$ and to apply the \emph{logical} Pauli operator $\overline{X}^{u}$ associated with $S$ to the noisy codeword itself, resulting in the desired transformation $E\ket{\overline{\psi}_{x\oplus u}}^S$.

To re-randomize the code and the error, we first observe that the shifted codeword $\ket{\overline{\psi}_{x\oplus u}}^S$ can be written as 
$$
\ket{\overline{\psi}_{x\oplus u}}^S = U_{\mathrm{Enc}}^S(\ket{0^{n-k}} \otimes \ket{x \oplus u})
$$
for some (not necessarily random) encoding Clifford $U_{\mathrm{Enc}}^S \in \Cliff_n$. Because $\Cliff_n$ forms a finite group, this suggests that one could simply sample a uniformly random $C \sim \Cliff_n$ and consider the state $CE\ket{\overline{\psi}_{x\oplus u}}^S = (CEC^\dag) C\ket{\overline{\psi}_{x\oplus u}}^S$, where $C\ket{\overline{\psi}_{x\oplus u}}^S$ now comes from a random stabilizer code for a uniformly random encoding Clifford $C \cdot U_{\mathrm{Enc}}^S$. While this does seem to result in a re-randomized stabilizer code, the aforementioned transformation could potentially \emph{blow up} the weight of the Pauli error $E$. In fact, it is well-known that random Cliffords are \emph{Pauli-mixing}~\cite{10.5555/3179473.3179474,aharonov2017interactiveproofsquantumcomputations}: they take any non-identity Pauli operator and map it to a uniformly random non-identity Pauli operator via conjugation. This seems to suggest that any naive worst-case to average-case reduction for LSN is doomed to fail: the weight of the re-randomized Pauli error now
appears to follow a Binomial distribution with parameter $3/4$, thereby inducing a high-noise distribution which is statistically far from any reasonable noise channel (e.g., such as a local depolarizing channel).

To overcome this barrier we develop a re-randomization strategy which is much more \emph{gentle} on the error (i.e., it does not cause it to blow up), and yet still ensures that the code gets somewhat re-randomized. Our strategy is to follow the random Pauli operator with a {\em twirl}---another random unitary consisting of a random permutation operator, followed by a layer of random single-qubit Cliffords. Similar ensembles of unitaries been used in the {\em randomized compiling and benchmarking} literature in order to tailor noise with arbitrary coherence and spatial correlations into a symmetric Pauli channel \cite{Wallman_2016, Emerson_2007}. To our knowledge, however, our work is the first to identify the twirl as a useful tool in the context of a worst-case to average-case reduction. Under this twirl, a worst-case error of weight $w$ is transformed into a uniformly random error of weight $w$; in particular, the weight of the error remains invariant. However, that the distributions of the error and the code are now \emph{correlated}, which requires a much more refined analysis.

\paragraph{Complexity of Learning Stabilizers with Noise.} 
What is the computational complexity of solving the LSN problem? Notice that the description of the learning task features quantum inputs, which means that its complexity cannot be characterized by traditional
complexity classes such as $\mathsf{BQP}$ or $\mathsf{QMA}$ which deal with classical-input decision problems. Instead, we show that (a variant of) the $\mathsf{LSN}$ problem lies in the complexity class called $\mathsf{avgUnitaryBQP}^{\mathsf{NP}}$, a (distributional and oracle) unitary synthesis class which was recently formalized by Bostanci et al.~\cite{bostanci2023unitarycomplexityuhlmanntransformation}.
At a high level, this follows from the quantum Gilbert-Varshamov bound which ensures that a random $\mathsf{LSN}$ instance carries a unique stabilizer syndrome with overwhelming probability. Once the syndrome is computed, it can then be decoded with the help of an $\mathsf{NP}$ oracle. In other words, any polynomial-time quantum machine which access to the oracle can \emph{synthesize} the appropriate Pauli correction for the noisy quantum codeword.

In \Cref{sec:complexity}, we give a brief review of unitary complexity and also prove our aforementioned complexity upper bound on the hardness of the $\mathsf{LSN}$ problem.

\subsection{Applications.}

\paragraph{Learning from quantum data.} Just as $\lpn$ has been fundamental to lower bounds in classical learning theory, we expect that $\lsn$ will be a useful tool for proving lower bounds in quantum learning theory. In \Cref{sec:learning} we identify one such learning setting: learning from quantum data \cite{Caro_2021,chung2021sampleefficientalgorithmslearning,FQR,Caro24}, a generalization of Probably Approximately Correct (PAC) learning to the quantum setting. Here, the goal is to learn a map  $\rho: \mathcal{X} \rightarrow L(\mathcal{H}_d)$ from classical to quantum data -- for example, a Hamiltonian can be construed as a map from temperatures to Gibbs states, or time-evolved states. 

In one special case of this task known as {\em learning state preparation processes}, the learner is allowed to observe input-output pairs for an unknown map, but the inputs are sampled from a distribution and are not identical. Refs. \cite{chung2021sampleefficientalgorithmslearning,FQR} gave sample-efficient algorithms for learning in this setting, thus showing that it is information theoretically possible. But we show that these algorithms can never be computationally efficient---assuming the hardness of our $\lsn$ assumption. While the lack of computational efficiency was implicit for the fully general learning setting due to a result of \cite{AGS21}, that result does not apply when there are physically natural restrictions on the concept class, such as learning processes that involve quantum noise. Our $\lsn$ hardness assumption fills this gap. 


\paragraph{Constructing quantum bit commitment schemes.}

In \Cref{sec:commitments}, we give a cryptographic application of LSN and show how to construct a statistically hiding and computationally binding quantum commitment scheme~\cite{yan2023general}. This is a fundamental cryptographic primitive that allows two parties (called a \emph{sender} and \emph{receiver}) to engage in a two-phase quantum communication protocol:
in the first phase (the ``commit phase''), the sender sends a commitment (i.e., a quantum register) to a bit $b$ to the receiver;  the \emph{hiding} property of a bit commitment scheme ensures that the receiver cannot decide the value of $b$ from the commitment alone.
In the second phase (the ``reveal phase''), the sender sends another quantum register to the receiver that allows the receiver to compute the value of $b$; the \emph{binding} property of commitments ensures that the sender can only reveal the correct value of $b$, i.e.\ if the sender sent a reveal register that was meant to convince the receiver it had committed to a different value of $b$, the receiver would detect this. 

\subsection{Related Work}

Random linear codes have been extensively studied in the field of coding theory~\cite{Gilbert1952ACO, Varshamov}.
The \emph{Learning Parity with Noise}
problem was first proposed in~\cite{BFKL93}.
Blum, Kalai and Wasserman~\cite{BKW03} gave an algorithm that solves LPN in time $O(2^{k/\log k})$.
Berlekamp, McEliece and van Tilborg~\cite{BMT78} showed that the worst-case (decisional) syndrome decoding task is $\mathsf{NP}$-complete. 
Brakerski, Lyubashevsky, Vaikuntanathan and Wichs~\cite{BLVW} gave a worst-case to average-case reduction and showed that $\lpn$ is at least as hard as (a mildly hard variant of) of the nearest codeword problem (NCP). In subsequent work, Ref.~\cite{10.1007/978-3-030-84252-9_16} later gave an improved worst-case to average-case reduction in the subexponentially-hard constant-noise regime.~\cite{BLVW} also showed that $\lpn$ is contained in $\mathsf{Search}\mathsf{BPP}^{\mathsf{SZK}}$, and thus unlikely to be $\mathsf{NP}$-hard.

Smith~\cite{Graeme_thesis} showed a quantum analog of the Gilbert-Varshamov bound using the notion of random \emph{stabilizer codes}. The worst-case hardness of decoding quantum stabilizer codes as a classical decoding task has been extensively studied, and was found to be
$\mathsf{NP}$-complete~\cite{HsiehLeGall11,KuoLu12} or $\#\mathsf{P}$-complete~\cite{IP15}---depending on the problem. The key insight in these results is that classical decoding essentially reduces to quantum decoding.
Ref.s \cite{HsiehLeGall11} and \cite{KuoLu12} use a one-to-one correspondence between stabilizer codes and classical linear codes to prove that quantum maximum-likelihood decoding is $\mathsf{NP}$-complete. Ref.s \cite{PoulinChung,IP15} go one step further, pointing out that decoding stabilizer codes should be even harder than decoding classical codes because of {\em error degeneracy} in the quantum setting, whereby multiple different errors can lead to the same syndrome. Based on this insight,~\cite{IP15} showed that quantum maximum-likelihood decoding (which accounts for error degeneracy is in fact $\#\mathsf{P}$-complete). The proof, once again, reduces from a classical problem: evaluating the weight-enumerator polynomial of a classical binary linear code. Ref \cite{KapshikarKundu23} also show that finding the minimum distance of a quantum code is $\mathsf{NP}$-hard due to a reduction from classical minimum distance decoding.

The fact that all aforementioned results about quantum decoding rely on classical complexity primitives underscores the need for a formalism that captures the inherent quantum nature of the decoding task. A recent work of Bostanci, Efron, Metger, Poremba, Qian and Yuen~\cite{bostanci2023unitarycomplexityuhlmanntransformation} characterized the complexity of decoding general quantum channels (which includes quantum error correction) using the language of of \emph{unitary synthesis} problems. Crucially, this implies to uniquely quantum problems that feature quantum inputs and outputs. While their results do not explicitly analyze random stabilizer codes, we make use of their formalism in order to describe the complexity of our average-case LSN problem. Another recent line of work also argues one should build quantum cryptography from inherently quantum, rather than classical cryptographic hardness assumptions~\cite{Kretschmer21,morimae2022quantum,bostanci2023unitarycomplexityuhlmanntransformation,metger2024simpleconstructionslineardepthtdesigns,cryptoeprint:2024/1639}, as they are potentially harder to break. 

Recent work of Grewal, Iyer, Kretschmer and Liang~\cite{grewal2024efficientlearningquantumstates,10.1145/3618260.3649738} studied the sample complexity of learning general stabilizer states, and even more general states that feature few $\mathsf{T}$-gates. However, their setting is, in some sense, orthogonal to ours. For example, in our (single-shot) LSN learning task, the description of the encoding Clifford is entirely public, and the apparent hardness of learning arises in the presence of noise, whereas in their setting the task is to determine the set of stabilizers (and encoding Clifford) from several identical (pristine) copies of the unknown state.

Gollakota and Liang~\cite{Gollakota_2022} gave lower bounds on the sample complexity of PAC-learning noisy stabilizer states in the Statistical Query (SQ) model. While they do connect the hardness of their learning task to LPN via a reduction, it does not consider the problem of learning random stabilizers as in our setting, and therefore does not resemble our average-case learning task in any meaningful way.

\subsection{Open Problems}

Our work raises a number of interesting open questions; in particular:

\begin{itemize}
    \item Can we place a special variant of LSN in MicroCrypt~\cite{morimae2022quantum}, whereby the hardness of the problem plausibly does not imply the existence of one-way functions?\footnote{Interestingly, the LSN assumption as currently defined is not a MicroCrypt assumption, as the stabilizer syndrome decoding problem can always be solved with the help of an $\mathsf{NP}$ oracle due to the non-degeneracy  of random stabilizer codes. We thank Kabir Tomer and Justin Raizes for this observation.}

    \item Can we reduce the LPN problem to the standard LSN problem, where the underlying stabilizer code is uniformly random?

    \item What is the largest $n$-qubit Clifford subgroup that takes low-weight Paulis to low-weight Paulis? Note that this could potentially allow one to obtain a stronger worst-case to average-case reduction for LSN.

    \item Can we prove a much better worst-case to average-case reduction altogether which applies to the standard LSN problem for random stabilizer codes?
    \item Can we prove a search-to-decision reduction for LSN, similar to what is known for both the $\lpn$ problem~\cite{10.1007/s00145-010-9061-2} and the LWE problem~\cite{regev2009lattices}? This could enable a number of other cryptographic primitives, such as (succinct) quantum encryption, directly under the LSN assumption.
\end{itemize}

\subsection*{Acknowledgements}
This material is based upon work supported by the Department of Energy, Office of Science, National Quantum Information
Science Research Centers, Quantum Systems Accelerator, under Grant
number DOE DE-SC0012704, and Co-design Center for Quantum Advantage (C2QA) under contract number DE-SC0012704.
The authors would like to thank Alexandru Gheorghiu, Anand Natarajan, Aram Harrow, Henry Yuen, John Bostanci, Jonas Haferkamp, Jonas Helsen, Jonathan Conrad, Soonwon Choi, Thomas Vidick and Vinod Vaikuntanathan for useful discussions.
The authors are supported by the National Science Foundation (NSF) under Grant No.\ CCF-1729369. 

\section{Preliminaries}\label{sec:prelims}

Let us first introduce some basic notation and relevant background.

\paragraph{Notation.}
For $N\in \mathbb{N}$, we use $[N] = \{1,2,\dots,N\}$ to denote the set of integers up to $N$. The symmetric group on $[N]$ is denoted by $\frak{S}_N$. 
In slight abuse of notation, we sometimes identify elements $x \in [N]$ with bit strings $x \in \bit^n$ via their binary representation whenever $N=2^n$ and $n \in \mathbb{N}$. Similarly, we identify permutations $\pi \in \frak{S}_N$ with permutations $\pi: \bit^{n} \rightarrow 
\bit^n$ over bit strings of length $n$.

We write $\negl(\cdot)$ to denote any \emph{negligible} function, which is a function $f$ such that, for every constant $c \in \mathbb{N}$, there exists an integer $N$ such that for all $n > N$, $f(n) < n^{-c}$.

\paragraph{Probability theory.} 

The notation $x \sim X$ describes that an element $x$ is drawn uniformly at random from the set $X$, and we use $\mathrm{Unif}(X)$ to denote the uniform distribution over $X$. Similarly, if $\algo D$ is a general distribution, we let
$x \sim \algo D$ denote sampling $x$ according to $\algo D$.
We denote the expectation value of a random variable $X$ 
by $\mathbb{E}[X] = \sum_{x} x \Pr[ X = x] $. If $\algo D$ is a distribution over a set $X$, then we denote by $\algo D^{\otimes n}$ the $n$-wise product distribution over the set $X \times \dots \times X$. For a parameter $p\in [0,1]$, we let $\Ber_p$ denote the Bernoulli distribution with 
$$
\Pr[X=1]=p \quad\text{ and }\quad \Pr[X=0]=1-p, \quad \text{ for } X \sim \Ber_p.$$
We let $\mathsf{Bin}_{n,p}$ denote the Binomial distribution with $\Pr[X=k] = \binom{n}{k}p^k (1-p)^{n-k}$, for a random variable $X \sim \mathsf{Bin}_{n,p}$.

\paragraph{Quantum information.}

 For a comprehensive background on quantum computation, we refer to~\cite{NC}. We denote a finite-dimensional complex Hilbert space by $\mathcal{H}$, and we use subscripts to distinguish between different systems (or registers). For example, we let $\mathcal{H}_{\reg A}$ be the Hilbert space corresponding to a system ${\reg A}$. 
The tensor product of two Hilbert spaces $\algo H_{\reg A}$ and $\algo H_{\reg B}$ is another Hilbert space denoted by $\algo H_{\reg{AB}} = \algo H_{\reg A} \otimes \algo H_{\reg B}$.
The Euclidean norm of a vector $\ket{\psi} \in \algo H$ over the finite-dimensional complex Hilbert space $\mathcal{H}$ is denoted as $\| \psi \| = \sqrt{\braket{\psi|\psi}}$. 
Let $\linear(\algo H)$
denote the set of linear operators over $\algo H$. A quantum system over the $2$-dimensional Hilbert space $\mathcal{H} = \mathbb{C}^2$ is called a \emph{qubit}. For $n \in \mathbb{N}$, we refer to quantum registers over the Hilbert space $\mathcal{H} = \big(\mathbb{C}^2\big)^{\otimes n}$ as $n$-qubit states. We use the word \emph{quantum state} to refer to both pure states (unit vectors $\ket{\psi} \in \mathcal{H}$) and density matrices $\rho \in \mathcal{D}(\mathcal{H)}$, where we use the notation $\mathcal{D}(\mathcal{H)}$ to refer to the space of positive semidefinite matrices of unit trace acting on $\algo H$.

A quantum channel $\Phi: \linear(\algo H_{\reg A}) \rightarrow \linear(\algo H_{\reg B})$ is a linear map between linear operators over the Hilbert spaces $\algo H_{\reg A}$ and $\algo H_{\reg B}$. Oftentimes, we use the compact notation $\Phi_{\reg{A \rightarrow B}}$ to denote a quantum channel between $\linear(\algo H_{\reg A})$ and $\linear(\algo H_{\reg B})$. We say that a channel $\Phi$ is \emph{completely positive} if, for a reference system $\reg R$ of arbitrary size, the induced map $\id_{\reg R} \otimes \Phi$ is positive, and we call it \emph{trace-preserving} if $\Tr{\Phi(X)} = \Tr{X}$, for all $X \in \linear(\algo H)$. A quantum channel that is both completely positive and trace-preserving is called a quantum $\mathsf{CPTP}$ channel.
A \emph{unitary} $U: \linear(\mathcal{H}_{\reg A}) \rightarrow \linear(\mathcal{H}_{\reg A})$ is a special case of a quantum channel that satisfies $U^\dagger U = U U^\dagger = \id_{\reg A}$. We denote the $n$-qubit unitary group by $\mathcal{U}_n$.
A \emph{projector} ${\Pi}$ is a Hermitian operator such that ${\Pi}^2 = {\Pi}$, and a \emph{projective measurement} is a collection of projectors $\{{\Pi}_i\}_i$ such that $\sum_i {\Pi}_i = \id$.
A positive-operator valued measure ($\mathsf{POVM}$) is a set of Hermitian positive semidefinite operators $\{M_i\}$ acting on a Hilbert space $\mathcal{H}$ such that $\sum_{i} M_i = \id$. 
A linear map $U \in \linear(\mathcal{H}_{\mathsf{A}}, \mathcal{H}_{\mathsf{B}})$ is called a partial isometry if there exists a projector $\Pi \in \linear(\mathcal{H}_{\mathsf{A}})$ and an isometry $\tilde U \in \linear(\mathcal{H}_{\mathsf{A}}, \mathcal{H}_{\mathsf{B}})$ such that $U = \tilde U \Pi$.
We call the image of the projector $\Pi$ the \emph{support} of the partial isometry $U$.
Because a partial isometry cannot be implemented in practice (it is not a trace-preserving operation), we also define a \emph{channel completion} of a partial isometry as any quantum channel that behaves like the partial isometry on its support, and can behave arbitrarily on the orthogonal complement of the support; specifically, for any
partial isometry $U \in \linear(\mathcal{H}_{\mathsf{A}}, \mathcal{H}_{\mathsf{B}})$, a \emph{channel completion of $U$} is a $\mathsf{CPTP}$ channel $\Phi \in \linear(\mathcal{H}_{\mathsf{A}}, \mathcal{H}_{\mathsf{B}})$ such that
\begin{align*}
    \Phi(\Pi \rho \Pi) = U \Pi \rho \Pi U^\dagger\,, \quad \text{ for } \rho \in \algo D(\mathcal{H}_{\mathsf{A}}),
\end{align*}
where $\Pi \in \linear(\mathcal{H}_{\mathsf{A}})$ is the projector onto the support of $U$.
If $\Phi$ is a unitary or isometric channel, we also call this a \emph{unitary} or \emph{isometric completion} of the partial isometry.

\paragraph{Quantum distance measures.}

Let $\rho,\sigma \in \algo D(\algo H)$ be two density matrices acting on the same Hilbert space $\algo H$. The (squared)  \emph{fidelity} between $\rho$ and $\sigma$ is defined as 
$$
\mathrm{F}(\rho,\sigma) = \| \sqrt{\rho} \sqrt{\sigma} \|_1^2 \, ,
$$
where $\| \cdot \|_1$ is the trace norm.
The \emph{trace distance} of $\rho,\sigma \in \mathcal{D}(\mathcal{H)}$ is given by
$$
\delta_{\mathsf{TD}}(\rho,\sigma) = \frac{1}{2} \| \rho - \sigma \|_1.
$$
The two distance measures are related via the Fuchs-van de Graaf inequalities:
\[
    1 - \sqrt{\mathrm{F}(\rho,\sigma)} \leq \delta_{\mathsf{TD}}(\rho,\sigma) \leq \sqrt{1 - \mathrm{F}(\rho,\sigma)}\,.
\]

We also use the following inequality.

\begin{lemma}[Strong convexity of trace distance (\cite{NC}, Theorem 9.3)]\label{lem:strong-convexity}
Let $\vec p=\left\{p_i\right\}$ and $\vec q=\left\{q_i\right\}$ be probability distributions over the same index set, and let $\{\rho_i\}$ and $\{\sigma_i\}$ be density operators, also with indices from the same index set. Then,
\begin{equation*}
\delta_{\mathsf{TD}}\left(\sum_i p_i \rho_i,\sum_i q_i \sigma_i \right) \leq \sum_i p_i \cdot \delta_{\mathsf{TD}}\left(\rho_i,\sigma_i\right)+\delta_{\mathsf{TV}}(\vec p,\vec q).
\end{equation*}
\end{lemma}

\paragraph{Uhlmann's theorem.} 

We frequently make use of the following theorem.

\begin{theorem}[{Uhlmann's theorem \cite{uhlmann1976transition}}] \label{thm:uhlmann}
Let $\ket{\psi}_{\mathsf{AB}}$ and $\ket{\phi}_{\mathsf{AB}}$ be pure states that live in a Hilbert space $\algo H_{\mathsf{AB}}$, and let $\rho_\mathsf{A}$ and $\sigma_\mathsf{A}$ denote their respective reduced states in register $\mathsf{A}$.
Then, there exists a unitary $U \in \mathrm{L}(\algo H_{\mathsf{B}})$ acting only on register $\mathsf{B}$ such that
\[
\mathrm{F}(\rho_{\mathsf{A}}, \sigma_{\mathsf{A}}) = |\bra{\phi}_{\mathsf{AB}} (\id_\mathsf{A} \otimes U_{\mathsf{B}}) \ket{\psi}_{\mathsf{AB}}|^2 \,.
\]
\end{theorem}

\paragraph{Gentle Measurement.} We also make use of the following well-known lemma, which is often called the Gentle Measurement Lemma.

\begin{lemma}[\cite{Wilde_2013}, Lemma 9.4.1]
Let $\rho \in \algo D(\algo H)$ be an arbitrary density matrix, and let $\Lambda$ be any positive semidefinite hermitian matrix. Then,
$$
\delta_{\mathsf{TD}}\left(\rho,\frac{\sqrt{\Lambda} \rho \sqrt{\Lambda}}{\mathrm{Tr}[\Lambda \rho]}\right)  \, \leq \, \sqrt{1 -\mathrm{Tr}[\Lambda \rho]}.
$$
\end{lemma}

\paragraph{Permutation operators.} Let $n \in \mathbb{N}$ be an integer. Then, for a permutation $\pi \in \frak{S}_n$, we define the corresponding $n$-qubit permutation operator $\mathcal{Q}(\pi)$ over $(\mathbb{C}^2)^{\otimes n}$ as
$$
\mathcal{Q}(\pi):=\sum_{i_1, \ldots, i_n\in\{0,1\}}\left|i_{\pi^{-1}(1)}, \ldots, i_{\pi^{-1}(n)}\right\rangle\left\langle i_1, \cdots, i_n\right|.$$
In other words, $\mathcal{Q}(\pi)$ is the unitary operator that permutes all of the single-qubit qubit registers according to the permutation $\pi$. By linearity, the operator $\mathcal{Q}(\pi)$ also permutes any product of single-qubit linear operators $\vec O_1,\dots,\vec O_n \in \mathrm{L}(\mathbb{C}^2)$ as follows:
$$
\mathcal{Q}(\pi) (\vec O_1 \otimes \dots \otimes \vec O_n) \mathcal{Q}(\pi)^\dag = (\vec O_{\pi^{-1}(1)} \otimes \dots \otimes \vec O_{\pi^{-1}(n)}).
$$

\subsection{Pauli group.} The $n$-qubit Pauli group $\Pn$ consists of $n$-fold tensor products of Pauli operators. In other words, this is a group of order $|\Pn| = 2^{2n+1}$ with
$$
\Pn = \{ \pm \id, \pm \mathsf{X}, \pm \mathsf{Y},\pm \mathsf{Z} \}^{\otimes n}.
$$
Once we switch to the symplectic representation, we are going to ignore signs and instead consider the quotient group given by
$$
\bar{\mathcal{P}}_n = \Pn/\{\pm \id^{\otimes n}\}.
$$

An important property of Pauli group is the fact that any two Pauli operators either commute or anti-commute, i.e., it holds that $PQ=\pm QP$ for any Pauli operators $P,Q$. 

\paragraph{Symplectic representation.} Instead of working with $\Pn$, we sometimes use the quotient group $\bar{\mathcal{P}}_n = \Pn/\{\pm \id^{\otimes n}\}$ to reason about Paulis. In this case, every unsigned Pauli $P \in \bar{\mathcal{P}}_n$ can be specified in terms of a pair $\vec p_x, \vec p_z \in \bit^n$ such that
$$
P = \bigotimes_{i=1}^n \mathsf{X}^{\vec p_{x,i}} \cdot \bigotimes_{i=1}^n \mathsf{Z}^{\vec p_{z,i}}.
$$
Therefore, we can directly identify the unsigned Paulis $\bar{\mathcal{P}}_n$ with a binary vector space of dimension $2n$ such that $\bar{\mathcal{P}}_n \cong \Z_2^{2n}$. Note that multplication of Paulis $P,Q$ in $\bar{\mathcal{P}}_n$, each represented by $\vec p =(\vec p_x|\vec p_z)$ and $\vec q=(\vec q_x|\vec q_z)$, amounts to addition in $\Z_2^{2n}$:
$$
(\vec p_x|\vec p_z) \cdot (\vec q_x|\vec q_z) = (-1)^{\vec p_x \cdot \vec q_z + \vec p_z \cdot \vec q_x} (\vec p_x \oplus \vec q_x|\vec p_z \oplus \vec p_x).
$$
We call $(\vec p_x|\vec p_z) \odot (\vec q_x|\vec q_z) := \vec p_x \cdot \vec q_z + \vec p_z \cdot \vec q_x$ the symplectic inner product. Two Paulis $P,Q$ commute if and only if their symplectic inner product of its $\Z_2^{2n}$ representations $(\vec p_x|\vec p_z)$ and $(\vec q_x|\vec q_z)$ vanishes. In other words, if and only if $\vec  p \odot \vec q = 0 \Mod{2}$.

\subsection{Clifford group.} The Clifford group is the set of unitaries that normalizes the Pauli group; that is
\begin{equation}
\Cliff_n=\left\{U \in \algo U_n \mid U P U^{\dagger} \in \algo{P}_n, \, \forall P \in \algo{P}_n\right\}
\end{equation}
The Clifford group also contains operators of the form $e^{i\theta} I$. for some $\theta \in [0, 2\pi]$. These global phase operators are often irrelevant to us because Clifford operators act by conjugation, and thus we consider two Clifford operators to be equivalent if they differ only by a global phase. Thus we will define
\begin{equation}
\vCliff_n= \Cliff_n/\{e^{i\theta}\mathcal{P}_n\}.
\end{equation}


We also quote a result about the generators of the Clifford group:
\begin{theorem}[Clifford generators \cite{gottesman_book}]
    Any gate in the Clifford group $\Cliff_n$ can be written as a product of $e^{i \theta} I, H_i, R_{\pi / 4, i}$, and $\mathrm{CNOT}_{i, j}$, with $i, j=1, \ldots, n$ and $\theta \in[0,2 \pi)$.
\end{theorem}
In the rest of this paper, we assume all Clifford circuits are written in terms of the above gates. In fact, we can compile the basic set of Clifford generators above into a more convenient gateset: 
\begin{equation}
    \{e^{i \theta} I, H, R_{\pi / 4}, \mathrm{CNOT}, \mathrm{SWAP}, \mathrm{C}\mbox{-}\mathrm{Z}.\}
\end{equation}
This is seen to be equivalent to the original gateset by using the identities $\mathrm{C}\mbox{-}\mathrm{Z}=(I \otimes H) \operatorname{CNOT}(I \otimes H)$ and $\mathrm{SWAP}_{i,j} = \mathrm{CNOT}_{i\rightarrow j}\otimes \mathrm{CNOT}_{j\rightarrow i} \otimes \mathrm{CNOT}_{i\rightarrow j}$, pictorially illustrated in Figure \ref{fig:transpositions}. Here we have, for the only time in this paper, used the notation $\mathrm{CNOT}_{i\rightarrow j}$ to denote a CNOT with $i$ as the control qubit and $j$ as the target qubit. 

\paragraph{Permutations and Local Cliffords.} 

Note that, for any $\pi \in \frak{S}_n$, the permutation operator $\mathcal{Q}(\pi)$ can be implemented as a Clifford operation by composing transpositions. These are themselves Clifford operations consisting of three consecutive CNOTs, as illustrated in \Cref{fig:transpositions}.
\begin{figure}[h]
    \centering
    \begin{quantikz}
       (i)\, & \qw & \ctrl{2} & \targ{} & \ctrl{2} & \qw& \qw \\
           & \vdots & & & & & \\
       (j)\, & \qw & \targ{} & \ctrl{-2} & \targ{} & \qw & \qw \\
    \end{quantikz}
    \caption{Circuit to transpose qubit $i$ and $j$ \label{fig:transpositions}}
\end{figure}
This motivates us to consider the $n$-qubit subgroup $\plc_n \leq \vCliff_n$ generated by permutation operators and local Clifford gates:
\begin{equation*}
 \plc_n := \left\{ C \in \vCliff_n \mid   C = \bigotimes_{i=1}^nC_i\circ \mathcal{Q}(\pi) \, : \,  C_1, \dots, C_n \in \vCliff_1 \,\,\text{ and } \,\, \pi \in \frak{S}_n  \right\}.
\end{equation*}

It is easy to see that $\mathsf{PLC}_n$ forms a group, as captured by the following lemma. 

\begin{lemma}[$\plc$ is a group\label{lem:plcgroup}]
    The set of all Clifford unitaries in $\mathsf{PLC}_n$ form a subgroup of the $n$-qubit unitary group under unitary composition (i.e., matrix multiplication).
\end{lemma}
This group is a key tool in our worst-to average-case reduction in \Cref{sec:worst-avg-reduction}.

\subsection{Quantum Noise Channels}

\paragraph{Single-qubit depolarizing noise.} For $M \in L(\mathbb{C}^d)$, we define the single-qubit depolarizing noise channel $\mathcal{D}_p: L(\mathbb{C}^d) \rightarrow L(\mathbb{C}^d)$ as one that acts as
\begin{equation}
    \mathcal{D}_p(M):=(1- 4/3p)M+ 4/3p\Tr{M}\frac{\mathbb{I}}{2},\quad p\in[0,3/4].
\end{equation}
This channel, when acting on a quantum state $\rho$, has the Kraus representation 
\begin{equation}\label{eq-single_qubit_depolar_Pauli}
\mathcal{D}_p(\rho) = \frac{p}{3} X\rho X + \frac{p}{3} Y \rho Y + \frac{p}{3} Z\rho Z + (1- p)\rho,
\end{equation}
Thus, a single-qubit depolarizing channel can be thought of as a random Pauli channel $E(\cdot)E^{\dagger}, \, E\in \mathcal{P}_1$, where $E$ can be sampled as follows: pick $wt(E) \sim \mathsf{Bernoulli}(p)$ (i.e. apply $\text{id}(\cdot)$ with probability $1-p$) then uniformly sample from the $E\in \mathcal{P}_1$ with weight $w$.

\paragraph{Tensor product of single-qubit depolarizing noise channels.}
We will model noise acting on an $n$-qubit quantum state as a tensor product of $n$ single-qubit depolarizing noise channels with parameter $p \in (0,3/4)$, where for $\rho\in L(\mathbb{C}^{2^n})$ 
\begin{equation}\label{eq:nqubitdep}
\mathcal{D}_p^{\otimes n}(\rho) := \sum_{E\in \mathcal{P}_n} \left(\frac{p}{3}\right)^{|E|}\left(1-p\right)^{n-|E|} E\rho E^{\dagger} =: \sum_{E\in \mathcal{P}_n} \Pr_{E\sim \mathcal{D}_p^{\otimes n}}[E] E\rho E^{\dagger}.
\end{equation}
It is not hard to check that \Cref{eq:nqubitdep} yields \Cref{eq-single_qubit_depolar_Pauli} for $n=1$. This channel is equivalent to acting with a random Pauli channel $E(\cdot)E^{\dagger}, \, E\in \mathcal{P}_n$, where $E$ can be sampled as follows: pick $wt(E) = w \sim \mathsf{Binom}(n, p)$, then uniformly sample from the $E\in \mathcal{P}_n$ with weight $w$. We will often use this interpretation of noise channels as probabilistically applying Pauli errors. 

\paragraph{Bit-flip noise channels.}
For $\rho \in L(\mathbb{C}^2)$, the single-qubit bit-flip noise channel $\mathcal{F}_p: L(\mathbb{C}^2) \rightarrow L(\mathbb{C}^2)$ acts as 
\begin{equation}\label{eq:bitflip}
    \mathcal{F}_p(M):= (1-p) \rho + p X\rho X.
\end{equation}
This may, as usual, be extended to a tensor product of $n$ single-qubit bit-flip noise channels via:
\begin{equation}\label{eq:nbitflip}
    \mathcal{F}_p^{\otimes n}(\rho):= \sum_{b\in \{0,1\}^n} p^{|b|}(1-p)^{n-|b|} X^b\rho X^{b},
\end{equation}
where for $i\in [n]$, the $i$-th qubit of $X^b$ is $(X^b)_i:= X^{\id(b_i=1)}.$


\subsection{Stabilizer Codes.}\label{sec:stabilizer}
We can use $\Pn$ to characterize a quantum error correcting code as follows. Let $S \leq \Pn$ be an abelian subgroup which has $n-k$ generators. In other words,
$$
S = \langle \vec g_1,\dots, \vec g_{n-k} \rangle.
$$
Then, we define the codespace $C(S)\subseteq (\mathbb{C}^2)^{\otimes n}$ as the simultaneous $(+1)$ eigenspace of all elements of $S$. \footnote{Because $S \subseteq \Pn$ is an abelian subgroup its elements can all be simultaneously diagonalized.}. In other words,
$$
C(S) = \big\{ \ket{\psi} \in \left(\mathbb{C}^2 \right)^{\otimes n} \, : \, M \ket{\psi} = (+1) \ket{\psi}, \,\, \forall M \in S\big\}.
$$
If $S$ has $n-k$ generators, then we know that $C(S)$ has precisely $2^k$ many codewords. We also use $\text{Stab}(n,k) $ to denote the set of all stabilizer subgroups of $\mathcal{P}_n$ with $n-k$ generators. 

The normalizer, denoted by $N(S)$, is the set of Paulis that commute with $S$\footnote{Technically, this is the centralizer but in our case they are the same.}:
$$
N(S) = \big\{ P \in \Pn \, : \, PQ = QP, \,\, \forall Q \in S\big\}.
$$
How do we detect errors? Suppose we have a class of errors $\mathcal{E} = \{ E_a\}$ where each $E_a \in \Pn$ is a low-weight Pauli error. Let $\ket{\psi} \in C(S)$ be a codeword. Then, for every generator $\vec g_i$ in $S$, one of two possible events takes place:
\begin{itemize}
    \item $\vec g_i$ commutes with $E_a$, in which case
    $$
\vec g_i E_a \ket{\psi}= E_a \vec g_i \ket{\psi} = E_a \ket{\psi}.
$$
\item $\vec g_i$ anti-commutes with $E_a$, in which case
    $$
\vec g_i E_a \ket{\psi}= (-1) E_a \vec g_i \ket{\psi} = (-1) E_a \ket{\psi}.
$$
\end{itemize}
We can detect this phase by performing a measurement. In general, we want to do this for every generator and collect a syndrome vector $s_a \in \bit^{n-k}$ such that
$$
\vec g_i E_a = (-1)^{s_{a,i}} E_a \vec g_i \, , \quad \forall i \in [n-k].
$$
The task of recovering a Pauli error from its corresponding stabilizer syndrome is called the stabilizer \emph{syndrome decoding problem}.

The {\em distance} of a quantum code $C$ is the minimum Hamming weight $d$ of an error that the code cannot correct. We use the following well-known fact (see e.g.~\cite{gottesman_book}).

\begin{theorem}[Knill-Laflamme conditions\label{thm:KLconditions}]
Let $\mathcal{E}$ be a set of errors with maximum Hamming weight $d$. Then, $C\subseteq (\mathbb{C}^2)^{\otimes n}$ is a $[[n,k,d]]$ quantum error correcting code with distance $d$ if and only if for every $\ket{\psi},\ket{\phi} \in C$ and for all $E_a, E_b \in \mathcal{E}$, it holds that
$$
\bra{\psi}E_a^{\dagger} E_b\ket{\phi}=c_{a b}\braket{\psi|\phi}
$$
for some hermitian matrix $c_{a b}$. 
Note that $c_{a b}$ does not depend on $\ket{\psi}$ or $\ket{\phi}$.
\end{theorem}

We say that a code is {\em non-degenerate} if $c_{ab}=\delta_{ab}.$ The interpretation of the non-degeneracy condition is that two different errors take any two states in the same subspace to orthogonal subspaces. 

\paragraph{Encoding circuits for stabilizers.} Given a description of some stabilizer code in terms of the generating elements of the stabilizer group $S = \langle \vec g_1,\dots,\vec g_{n-k} \rangle$, there exists a Clifford circuit $U_{S} \in \Cliff_n$ to encode any initial $k$-qubit state $\ket{\psi} \in (\mathbb{C}^2)^{\otimes k}$ as a state inside the subspace stabilized by $S$. Let
$$
\ket{\bar{\psi}}^S := \UES\left(\ket{0^{n-k}} \otimes \ket{\psi} \right)\in C(S).
$$
denote the encoding of $\ket{\psi}$ in the subspace stabilized by $S$. 

How can we find $\UES$? Note that the initial state $\ket{0^{n-k}} \otimes \ket{\psi}$ itself is also a codeword which is stabilized by the trivial stabilizer code $\langle Z_1,\dots,Z_{n-k}\rangle$, i.e.,
$$
\begin{matrix}
Z & I & I & \cdots & I & I & \cdots & I\\
I & Z & I & \cdots & I & I &\cdots & I\\
\vdots & \vdots & \vdots & \vdots &\ddots & \vdots & \ddots & \vdots \\
I & I & I & \cdots & Z & I &\cdots & I
\end{matrix}
$$
Any encoding unitary $\UES$ also acts by conjugation to map the initial set of stabilizers $Z_i$ to the new set of stabilizers specified by $S$, i.e. $\vec g_i = \UES Z_i {\UES}^{\dag}$. This is because for every $i \in [n-k]$, we have
\begin{align*}
\UES \left(\ket{0^{n-k}} \otimes \ket{\psi} \right) &= \UES \, Z_i \left(\ket{0^{n-k}} \otimes \ket{\psi} \right)\\
&= \underbrace{\left(\UES Z_i {\UES}^{\dag}\right)}_{= \vec g_i} \,\underbrace{{\UES}  \left(\ket{0^{n-k}} \otimes \ket{\psi} \right)}_{= \ket{\bar{\psi}}^S}.
\end{align*}
\begin{theorem}[Efficient and efficiently-findable encoding circuits for stabilizer codes]\label{thm:encode_stab}
Given a description of $S\in \Stab(n,k)$ there is an $O(n^3)$ time classical algorithm to write down an encoding circuit for $S$ with $O(n^2)$ Clifford gates.
\end{theorem}
\begin{proof}
The proof is presented in Gottesman (Section 6.4.1 of \cite{gottesman_book}) but we write it concisely.
Because of the remark before this theorem, it suffices to specify a circuit $\mathcal{C}$ that maps the final generators $\vec g_1,\dots,\vec g_{n-k}$ to the initial generators $Z_1,\ldots Z_{n-k}$, in the sense that $\mathcal{C}\vec{g}_i\mathcal{C}^{\dagger} = Z_i.$ Then we may output $\UES = \mathcal{C}^{\dag}.$ To find $\mathcal{C},$ we start by writing the generators $\vec g_1,\dots,\vec g_{n-k}$ in symplectic notation, representing them as two matrices $A,C\in \mathbb{Z}_2^{n\times n-k}$ where the $i$-th column of $A$ is $(\vec g_i)_x$ and the $i$-th column of $C$ is $(\vec g_i)_z$.  To find $\mathcal{C}$, we make use of the fact that conjugating a Pauli $P\in \bar{\mathcal{P}}_n$ with a Clifford $\mathcal{C}$ with symplectic representation $M_{\mathcal{C}}$ transforms the symplectic representation of $P$ (represented as a column vector) by left multiplication, as
\begin{equation}
    (p_x,p_z)^T\rightarrow M_{\mathcal{C}} (p_x,p_z)^T.
\end{equation}
Working exclusively within the symplectic picture, the goal is then to find a matrix $M_{\mathcal{C}}$ representing a valid Clifford such that 
\begin{equation}
    M_{\mathcal{C}}\binom{A}{\overline{\,C\,}} = \binom{0}{\overline{\id_{n,n-k}}}
\end{equation}
where $\id_{n,n-k}$ denotes a $n$ by $n-k$ matrix with $\id_{n-k}$ in its first $n-k$ rows and $0$s in all other locations. One can verify that the right-hand-side of the above equation represents the stabilizer $\langle Z_1,\ldots Z_{n-k} \rangle$ and hence the target. 

By checking Equations 6.79-6.82 of \cite{gottesman_book}, which give the symplectic representations of each Clifford generator, one can verify that row and column reductions on $\binom{A}{\overline{\,C\,}}$ can be implemented in symplectic space by acting with Clifford gates on the actual generators. Specifically, column reduction on $A$ composes the operation of adding the topmost row with a leading 1 to a different row which has an undesired leading 1. This corresponds to acting with a single $\mathrm{CNOT}$. Similarly, column reduction on $C$ can be implemented by acting with a single $R_{\pi/4}$ or $C-Z$. Row reduction on $A$ (via adding the column with a leading $1$ to other columns) corresponds to multiplication of generators, which effectively does not change the stabilizer group. Since only $O(n^2)$ additions of rows are needed to reduce $\binom{A}{\overline{\,C\,}}$ to the desired form, and each addition corresponds to adding $O(1)$ gates to the circuit, the row-reduction portion of the algorithm can be carried out in $O(n^2)$ time. 

There's one problem: in symplectic space phases do not exist, so the circuit $\mathcal{C}'$ resulting from the above procedure may not produce $S = \langle \vec g_1,\dots,\vec g_{n-k} \rangle$ with the proper signs. So we need to perform some final corrections. We then need to compute which $\vec g_i$s have the wrong signs, which by the following Lemma \ref{lem:sim} takes time $O(n^2)$ for each $g_i.$ For every $\vec{g}_i$ that has the wrong sign, we modify the circuit $\mathcal{C}'$ by acting with $X_i$ at the start of the circuit. This correction operation does not change the circuit complexity asymptotically and increases the classical runtime to $O(n^3)$, though we did not optimize this.
\end{proof}

We also quote a well-known result on simulating Cliffords.
\begin{lemma}[Simulating Clifford circuits\label{lem:sim}]
    Given a description of a Clifford circuit $C\in \Cliff_n$ where $C = \prod_{i=1}^m U_i$ and each $U_i$ is a 2-qubit gate from some generating set of $\vCliff_n$, for any Pauli $P\in \calP_n$ we may compute $CPC^{\dagger}$ in time linear in $O(m+n)$. 
\end{lemma}
\begin{proof}
    This is a special case of the Gottesman-Knill theorem \cite{gottesman_book}, which says that Clifford circuits acting on an initial $n$-qubit stabilizer state, followed by a sequence of $m$ Clifford group operations and Pauli measurements, can be efficiently classically simulated. The simulation algorithm is to keep track of how each operation transforms the stabilizer group of the initial state. 

    More concretely, in this case we have only one Pauli to keep track of, so we may compute $CPC^{\dagger}$ by computing the effect of each gate $U_i$ in sequence. Since each gate acts on a constant number of qubits, the time needed for the simulation and to write down the final Pauli scales as $O(m+n).$ Lookup tables for the action of gates from popular generating gatesets can be found in, e.g. Table 6.1 of \cite{gottesman_book}. 
\end{proof}

\paragraph{Random stabilizer codes.}
A random $[[n,k]]$ stabilizer code is a uniformly random choice of abelian subgroup $S \leq \mathcal{P}_n$ with $n-k$ generators. Note, here the Pauli signs are important and will determine what subspace is stabilized by $S$. Because the members of $S$ must commute, choosing $n-k$ elements uniformly from $\mathcal{P}_n$ will not always give a valid stabilizer code. We now give a rigorous proof of the fact that a random element of the Clifford group $\Cliff_n$, acting on any initial choice of $S$, generates a uniformly random $S$, and hence a uniformly random $[[n,k]]$ stabilizer code.


\begin{theorem}\label{thm:random-Stab}
Let $n \in \mathbb{N}$ be an integer and let $S = \langle \vec g_1,\dots,\vec g_{n-k} \rangle$ be any stabilizer with generators $\vec g_1,\dots,\vec g_{n-k} \in \mathcal{P}_n$. Then, the conjugated stabilizer code
$$
U S U^\dag = \langle U\vec g_1 U^\dag,\dots, U\vec g_{n-k} U^\dag \rangle, \quad \text{ for } U \sim \Cliff_n,
$$
yields a uniformly stabilizer in the set $\mathsf{Stab}(n,k)$.
\end{theorem}
\begin{proof}
First, we show that the Clifford group $\Cliff_n$ acts transitively on the set of stabilizers $\mathsf{Stab}(n,k)$. Let $S = \langle \vec g_1,\dots,\vec g_{n-k} \rangle$ be an arbitrary stabilizer with $n-k$ generators. From \cite{gottesman_book}, we know that there exists a Clifford operator $C \in \Cliff_n$ and a Pauli $P\in  \bar{\mathcal{P}}_n$ such that the composition of the two operations maps $S$ to the canonical stabilizer $Z$. In particular, we can let $V = PC$ such that
$$
V S V^\dag =  \langle V\vec g_1V^\dag,\dots, V\vec g_{n-k} V^\dag \rangle = \langle Z_1,\dots,Z_{n-k}\rangle.
$$
Likewise, from \cite{gottesman_book}, we also know that once we have the canonical all-Z stabilizer $Z=\langle Z_1,\dots,Z_{n-k}\rangle$, we can obtain any other stabilizer $S' = \langle \vec g_1',\dots,\vec g'_{n-k} \rangle$ via some other composition of operators $W = D Q$, where $D \in \Cliff_n$ and $Q\in  \bar{\mathcal{P}}_n$, i.e., 
$$
S' = \langle \vec g_1',\dots,\vec g'_{n-k} \rangle = \langle WZ_1W^\dag,\dots, W Z_{n-k} W^\dag \rangle.
$$
Therefore, for any pair of distinct stabilizers $S,S' \in \mathsf{Stab}(n,k)$ there exists an operator  $WV$ that maps $S$ to $S'$.  By using the fact that Cliffords are the normalizer of the Pauli group, $WV$ can be realized as a single Pauli operation followed by a single Clifford operation. 

Finally, we show that the probability that a random Clifford applied to an arbitrary stabilizer $S = \langle \vec g_1,\dots,\vec g_{n-k} \rangle$ yields any pair of distinct stabilizers $S_1,S_2$ with exactly the same probability. From before, there exists a Clifford $C \in \Cliff_n$ such that
\begin{align*}
\underset{U \sim \Cliff_n}\Pr[U S U^\dag = S_1]    &= \underset{U \sim \Cliff_n}\Pr[ (C U) S (CU)^\dag = C S_1 C^\dag] \\
&= \underset{U \sim \Cliff_n}\Pr[  (C U) S (CU)^\dag = S_2] \\
&= \underset{U \sim \Cliff_n}\Pr[  U S U^\dag = S_2].
\end{align*}
The last line follows from the fact that $\Cliff_n$ is a group, and hence the uniform distribution over $\Cliff_n$ is Clifford invariant.
\end{proof}

\paragraph{Quantum Gilbert-Varshamov bound.}

In a nutshell, the quantum Gilbert-Varshamov bound \cite{Graeme_thesis,gottesman_book} tells us that a random stabilizer code is both non-degenerate and has a good distance with high probability. This is captured by the following result.

\begin{theorem}[Quantum Gilbert-Varshamov bound,~\cite{Graeme_thesis}\label{thm:GVbound}]
Random stabilizer codes which are specified by a random stabilizer subgroup $S \sim \mathsf{Stab}(n,k)$ are non-degenerate and have distance $d$ with probability at least $1-d\cdot 2^{n H(d/n) } \cdot 3^{d} \cdot 2^{-n+k}$.
\end{theorem}

The above statement is perhaps best viewed through the lens of the Knill-Laflamme error correction conditions (Theorem \ref{thm:KLconditions}). Suppose that $S = \langle \vec g_1,\dots,\vec g_{n-k}\rangle$ is a non-degenerate stabilizer code with distance $d=2t+1$. Define the set $$
\mathcal{E}^{(t)} = \{ E_1^\dag E_2 \in \mathcal{P}_n \, : \, |E_1|,|E_2| \leq t\}
$$
which consists of weight-$t$ products of Pauli errors. Then, for all pairs of codewords $\ket{\overline{\psi_x}}, \ket{\overline{\psi_y}} \in C(S)$ with $x \neq y$, and for all $E_a^\dag E_b\in \mathcal{E}^{(t)}$, it holds that
$$
\bra{\overline{\psi_x}} E_a^\dag E_b \ket{\overline{\psi_y}}  =0.
$$
One way to see this is as follows. Suppose that $E_a^\dag E_b \notin N(S)$, then there must exist a generator $\vec g_i$ in $S = \langle \vec g_1,\dots,\vec g_{n-k}\rangle$ which anti-commutes with $E_a^\dag E_b$, and thus
\begin{align}\label{eq:anti-commutation-stab}
\vec g_i  E_a^\dag E_b \ket{\overline{\psi_x}}^S = -  E_a^\dag E_b \vec g_i \ket{\overline{\psi_x}}^S = - E_a^\dag E_b \ket{\overline{\psi_x}}^S.  
\end{align}
Using that $\vec g_i^2 = I^{\otimes n}$ together with \Cref{eq:anti-commutation-stab}, this implies that, for $x \neq y$,
\begin{equation}\label{eq:orthogonality}
\bra{\overline{\psi_x}} E_a^\dag E_b \ket{\overline{\psi_y}} = \bra{\overline{\psi_x}} \vec g_i E_a^\dag E_b \vec g_i\ket{\overline{\psi_y}} = -\bra{\overline{\psi_x}} E_a^\dag E_b \ket{\overline{\psi_y}}  =0.
\end{equation}

\section{The Learning Stabilizers with Noise problem}\label{sec:LPN}

In this section, we formally define the \emph{Learning Stabilizers with Noise} (LSN) problem as the natural quantum analog of the $\lpn$ problem. We begin with a set of definitions for the problem (and its variants), and then show that the problem is well-defined (i.e., it admits a unique solution) for appropriate choices of parameters.

\subsection{Definition}

We now provide a formal definition of our learning task.

\begin{definition}[Learning Stabilizers with Noise problem]\label{def:LSN}
Let $k \in \mathbb{N}$ be the security parameter and let $n=\poly(k)$ be an integer. Let $p \in (0,1/2)$ be a parameter. The Learning Stabilizers with Noise ($\mathsf{LSN}_{n,k,\mathcal{D}_p^{\otimes n}}$) problem is to find $x \in \bit^k$ given as input a sample
$$
\big(S\in \Stab(n,k), E \ket{\overline{\psi_x}}^S \big) \, ,
$$
where $S\sim \Stab(n,k)$ is a uniformly random stabilizer (specified in terms of a classical description of $S$), $E \sim \mathcal{D}_p^{\otimes n}$ is a Pauli error with $E \in \bar{\mathcal{P}}_n$,  $x \sim \bit^k$ is a random string
, and $\ket{\overline{\psi_x}}^S \in C(S)$ is the codeword 
\begin{equation*}
\ket{\overline{\psi_x}}^S:= U_{\mathrm{Enc}}^S(\ket{0^{n-k}} \otimes \ket{x})
\end{equation*}
for some canonical encoding circuit $U_{\mathrm{Enc}}^S$ for the stabilizer code associated with $S$. 
We say that a quantum algorithm solves the $\mathsf{LSN}_{n,k,\mathcal{D}_p^{\otimes n}}$ problem if it runs in time $\poly(k)$ and succeeds with probability at least $1/\poly(k)$ over the choice of $S, E$ and $x$, and its internal randomness. 
\end{definition}

Let us first state a few remarks.

\begin{remark}[Density matrix formulation] In \Cref{def:LSN}, the input to the learning algorithm is stated in the form $(S\in \Stab(n,k), E \ket{\overline{\psi_x}}^S)$, 
where $S\sim \Stab(n,k)$, $E \sim \mathcal{D}_p^{\otimes n}$ and $x \sim \bit^k$ is a random string. The pure state $E \ket{\overline{\psi_x}}^S$, however, should rather be understood as a density matrix of the form
$$
\rho_x^S = \mathcal{D}_{p}^{\otimes n}(\proj{\overline{\psi_x}}^S)= \sum_{E \in \bar{\mathcal{P}}_n} \Pr_{E \sim \mathcal{D}_{p}^{\otimes n}}[E] \cdot E \proj{\overline{\psi_x}}^S  E^\dag \, ,
$$
where we think of $\mathcal{D}_{p}^{\otimes n}$ as a product of local depolarizing channels with parameter $p$.
    
\end{remark}

\begin{remark}[Learning vs. decoding]
    The name `Learning Stabilizers with Noise' was chosen to parallel the name of the classical hardness assumption `Learning Parities with Noise'. The word `learning' here should be understood in the sense of Quantum Probably Approximately Correct (PAC) learning, where the example(s) seen by the learner/solver for $\lsn$ is a/are random quantum state(s). Notably, unlike in the setting of tomography, the quantum learner does not have multiple identical copies of the same quantum state. 
\end{remark}

\begin{remark}[Clifford representation]\label{remark:clifford-rep}
Recall that the input of the learner in \Cref{def:LSN} consists of a random stabilizer $S \in \Stab(n,k)$. In \Cref{thm:random-Stab}, we showed that random stabilizer codes can be equivalently described by uniformly random Clifford encoding circuits. Therefore, we can alternatively think of an $\mathsf{LSN}_{n,k,\mathcal{D}_p^{\otimes n}}$ instance as
$$
\left(C \in \Cliff_n, \, E \, C (\ket{0^{n-k}} \otimes \ket{x})\right)  
$$
where $C \sim \Cliff_n$ is a random $n$-qubit Clifford operator and the first argument refers to a classical description of $C$, $E \sim \mathcal{D}_p^{\otimes n}$ is an $n$-qubit Pauli error, and $x \sim \bit^k$ is a random string.
\end{remark}

\begin{remark}[Worst-case vs. average-case]
The computational task of solving $\mathsf{LSN}_{n,k,\mathcal{D}_p^{\otimes n}}$ in \Cref{def:LSN} is an average-case problem in the sense that the success probability of any algorithm is measured on average over the choice of stabilizer code $S$, error $E$ and secret $x$. We also consider the worst-case variant, where $S$ and $x$ and $E$ are not assumed to be random and instead chosen adversarially. We say that an algorithm solves the problem in the worst-case with probability $\delta >0$, if the algorithm succeeds at finding $x$ with probability $\delta$ when challenged on any instance of the problem---even for adversarial choices of stabilizer code $S$, error $E$ and secret $x$.
\end{remark}

\begin{remark}[Polynomial vs. quasi-polynomial hardness]\label{remark:quasi} In \Cref{def:LSN}, we considered the "polynomial hardness" of the LSN problem, i.e., we assume that a successful solver must run in time $\poly(k)$ and succeed with probability at least $1/\poly(k)$. In \Cref{sec:worst-avg-reduction}, we also establish a connection to "quasi-polynomial hardness" of the problem, where we assume that a successful solver must run in quasi-polynomial-time and must succeed with inverse-quasi-polynomial probability in $k$. 
\end{remark}

\paragraph{General variants.} Recall that the input in \Cref{def:LSN} consists of a sample
$$
\big(S, E \ket{\overline{\psi_x}}^S \big) \,\,\sim \,\,\mathsf{LSN}_{n,k,\mathcal{D}_p^{\otimes n}}
$$
where $S\sim \Stab(n,k)$ is a (classical description of) uniformly random stabilizer, where $E \sim \mathcal{D}_p^{\otimes n}$ is a Pauli error $E \in \bar{\mathcal{P}}_n$, where $x \sim \bit^k$ is a random string. Occasionally, we also consider more general variants of the problem, denoted by $\lsn_{n,k,\mathcal{N},\mathcal{S},\mathcal{I}}$, that feature samples
$$
\big(S, E \ket{\overline{\psi_x}}^S \big) \,\,\sim \,\,\lsn_{n,k,\mathcal{N},\mathcal{S},\mathcal{I}}
$$
which are captured by the following set of distributions (which depend on $n$ and $k$):
\begin{itemize}
    \item $\mathcal{N}$ is a \emph{general noise distribution} with support over $n$-qubit Pauli errors $E \in \bar{\mathcal{P}}_n$.

    \item $\mathcal{S}$ is a \emph{general distribution over Stabilizer codes}, either with support over the set of stabilizers $S \in \Stab(n,k)$, or over Clifford encoding circuits in $C\in\Cliff_n$.

    \item $\mathcal{I}$ is a \emph{general distribution over input strings} of the form $x \in \bit^k$.
\end{itemize}
Depending on the choice of distributions $\mathcal{N}$, $\mathcal{S}$ and $\mathcal{I}$, the learning task may either become easier or harder than the canonical learning problem in \Cref{def:LSN}.

\subsection{Existence of Unique Solutions}
\label{sec:unique-sol}

We now investigate for which parameter regime the learning problem $\mathsf{LSN}_{n,k,\mathcal{D}_p^{\otimes n}}$ in \Cref{def:LSN} is well-defined and can be solved information-theoretically. 
In anticipation of \Cref{sec:complexity} and \Cref{sec:commitments}, we show the existence of a unique solution by phrasing the $\mathsf{LSN}_{n,k,\mathcal{D}_p^{\otimes n}}$ problem as an \emph{Uhlmann transformation problem}. This can be understood as the problem of synthesizing the unitary in Uhlmann's theorem (\Cref{thm:uhlmann}). 

\paragraph{Uhlmann transformations and the LSN problem.}

To argue that the $\mathsf{LSN}_{n,k,\mathcal{D}_p^{\otimes n}}$ problem has a unique solution, we appeal to Uhlmann's theorem.
Suppose we are given as input an instance with repsect to the Clifford representation, 
$$
\big(C \in \Cliff_n, E \, C (\ket{0^{n-k}} \otimes \ket{ x})\big)  
$$
where $C \sim  \Cliff_n$ is a random $n$-qubit Clifford, $E \sim \mathcal{D}_p^{\otimes n}$ is an $n$-qubit Pauli error from a local depolarizing channel, and $x \sim \bit^k$ is a random $k$-bit secret. We can represent the density matrix corresponding to the quantum part of the instance\footnote{We also append an ancilla register in the state $\ket{0}$ for convenience.} as the result of discarding register $\mathsf{A}$ of the purification,
\begin{align}\label{eq:Q0}
\ket{Q^0}_{\mathsf{AB}} &= 
\sqrt{2^{-k}}\sum_{x} \sum_{E_a} \sqrt{\Pr_{E_a \sim \mathcal{D}_{p}^{\otimes n} }[E_a]} \,\Big(\ket{x} \otimes \ket{a}\Big)_{\mathsf{A}} \otimes \left(E_{a} \, C (\ket{0^{n-k}} \otimes \ket{x}) \otimes \ket{0}\right)_{\mathsf{B}}.    
\end{align}
For the sake of the proof, we also consider the following bipartite state given by
\begin{align}\label{eq:Q1}
\ket{Q^1}_{\mathsf{AB}} &= \sqrt{2^{-k}}\sum_{x} \sum_{E_a} \sqrt{\Pr_{E_a \sim \mathcal{D}_{p}^{\otimes n} }[E_a]} \,\Big(\ket{x} \otimes \ket{a}\Big)_{\mathsf{A}} \otimes \left(\ket{0^{n-k}} \otimes \ket{x} \otimes \ket{a}\right)_{\mathsf{B}}.    
\end{align}
 Using the quantum Gilbert-Varshamov bound, we can argue that the fidelity between the reduced states $Q_{\mathsf{A}}^0$ and $Q_{\mathsf{A}}^1$ on register $\mathsf{A}$ is near maximal---provided that $n$ is slightly larger than $k$, and that $p \in (0,1/2)$ is a sufficiently small constant. Therefore, by Uhlmann's theorem, there exists a unitary $U \in \mathrm{L}(\algo H_{\mathsf{B}})$ which acts on the $\mathsf{B}$ register and maps $\ket{Q^0}_{\mathsf{AB}}$ to another state which has near maximal overlap with $\ket{Q^1}_{\mathsf{AB}}$. In other words, the Uhlmann unitary $U$ allows us to solve $\mathsf{LSN}_{n,k,\mathcal{D}_p^{\otimes n}}$ and to recover the $x$ with overwhelming probability---thereby proving uniqueness.

To make this explicit, we first show the following technical lemma.

\begin{lemma}\label{lem:fidelity-lemma}
Let $n,k \in \mathbb{N}$ and let $\ket{Q^0}_{\mathsf{AB}}$ and $\ket{Q^1}_{\mathsf{AB}}$
be the bipartite states in Eq.~\eqref{eq:Q0} and Eq.~\eqref{eq:Q1}, respectively, for some $p \in (0,1/2)$ and Clifford $C \in \Cliff_n$. Then, with probability at least $ 1-3 np \cdot 2^{n H(3p) } \cdot 3^{3np} \cdot 2^{-n+k}$ over the choice of the random Clifford $C$, it holds that
$$
\mathrm{F}(Q_{\mathsf{A}}^0,Q_{\mathsf{A}}^1) \geq 1-4 \cdot e^{-\frac{np}{24}}.
$$
\end{lemma}
\begin{proof}
Before we bound the fidelity, we first make a couple of observations.
Define the projector onto system $\mathsf{A}$ given by
\begin{align}
\boldsymbol{\Pi}_{n,p} = \id_k \otimes \sum_{a \, : \, |E_a| \leq \frac{3}{2}np} \proj{a}.
\end{align}
Note that $\boldsymbol{\Pi}_{n,p}$ projects onto the support of all $n$-qubit Pauli operators of weight at most $\frac{3}{2}np$. For convenience, we also define the reduced density matrices
\begin{align*}
\hat{Q}^b_{\mathsf{A}} :=\frac{\mathrm{Tr}_{\mathsf{B}}\left[(\boldsymbol{\Pi}_{n,p} \otimes \id_{\mathsf{B}})Q^b_{\mathsf{AB}}(\boldsymbol{\Pi}_{n,p} \otimes \id_{\mathsf{B}})\right]}{\mathrm{Tr}[(\boldsymbol{\Pi}_{n,p} \otimes \id_{\mathsf{B}})Q^b_{\mathsf{AB}}(\boldsymbol{\Pi}_{n,p} \otimes \id_{\mathsf{B}})]}, \quad\quad \text{ for } b \in \bit.
\end{align*}

First, we observe the following about the reduced state $Q^0_{\mathsf{A}} = \mathrm{Tr}_{\mathsf{B}}\left[Q^0_{\mathsf{AB}}\right]$:
\begin{align*}
&\delta_{\mathsf{TD}}\left(\mathrm{Tr}_{\mathsf{B}}\left[Q^0_{\mathsf{AB}}\right], \hat{Q}^0_{\mathsf{A}} \right) \\
&\leq \delta_{\mathsf{TD}}\left(Q^0_{\mathsf{AB}}, \frac{(\boldsymbol{\Pi}_{n,p} \otimes \id_{\mathsf{B}})Q^0_{\mathsf{AB}}(\boldsymbol{\Pi}_{n,p} \otimes \id_{\mathsf{B}})}{\mathrm{Tr}[(\boldsymbol{\Pi}_{n,p} \otimes \id_{\mathsf{B}})Q^0_{\mathsf{AB}}(\boldsymbol{\Pi}_{n,p} \otimes \id_{\mathsf{B}})]} \right) && \quad (\text{Monotinicity of $\delta_{\mathsf{TD}}$})\\
&\leq  \sqrt{1- \mathrm{Tr}[(\boldsymbol{\Pi}_{n,p} \otimes \id_{\mathsf{B}})Q^0_{\mathsf{AB}}]} && (\text{Gentle measurement})\\
&=   \sqrt{\Pr_{E \sim \mathcal{D}_{p}^{\otimes n} }\left[|E| > \frac{3}{2}np\right]}\\
&\leq  \exp\left(-\frac{np}{24} \right). && (\text{Chernoff bound})
\end{align*}
Next, we make the following observation. By the quantum Gilbert-Varshamov, it follows that a random stabilizer code is non-degenerate and has distance at least $d=3np+1$ with probability at least $ 1-3 np \cdot 2^{n H(3p) } \cdot 3^{3np} \cdot 2^{-n+k}$ over the choice of encoding Clifford $C$, in which case for any pair of codewords with $x,y \in \bit^k$:
\begin{align}
\bra{\overline{\psi_x}} E_a^\dag E_b \ket{\overline{\psi_y}}  =0
\end{align}
by the Knill-Laflamme conditions---provided the errors $E_a, E_b$ have weight at most $|E_a|, |E_b| \leq \frac{3}{2}np$. Thus,
the aforementioned reduced states must be identical:
\begin{align}\label{eq:knill-lafrlamme-revisited}
\hat{Q}^0_{\mathsf{A}} = \hat{Q}^1_{\mathsf{A}}.   
\end{align}
To complete the proof, we now make another observation about the reduced state $Q^1_{\mathsf{A}} = \mathrm{Tr}_{\mathsf{B}}\left[Q^1_{\mathsf{AB}}\right]$. Using a similar approach as before, we obtain the bound
\begin{align*}
\delta_{\mathsf{TD}}\left( \hat{Q}^0_{\mathsf{A}},Q^1_{\mathsf{A}}\right)
&\leq \delta_{\mathsf{TD}}\left( \hat{Q}^0_{\mathsf{A}},\hat{Q}^1_{\mathsf{A}} \right) + \delta_{\mathsf{TD}}\left( \hat{Q}^1_{\mathsf{A}},Q^1_{\mathsf{A}} \right) && (\text{Triangle ineq.})\\
&\leq \, 0 \, + \, \delta_{\mathsf{TD}}\left( \hat{Q}^1_{\mathsf{A}},Q^1_{\mathsf{A}} \right) &&(\text{By Equation~\eqref{eq:knill-lafrlamme-revisited}})\\
&\leq  \sqrt{1- \mathrm{Tr}[(\boldsymbol{\Pi}_{n,p} \otimes \id_{\mathsf{B}})Q^0_{\mathsf{AB}}]} && (\text{Gentle measurement})\\
&=  \sqrt{\Pr_{E \sim \mathcal{D}_{p}^{\otimes n} }\left[|E| > \frac{3}{2}np\right]}\\
&\leq  \exp\left(-\frac{np}{24} \right). && (\text{Chernoff bound})
\end{align*}
Putting everything together, we can now apply the Fuchs-van de Graaf inequality, followed by the triangle inequality, followed by Bernoulli's inequality, to lower bound the fidelity between the reduced states as follows:
\begin{align*}
\mathrm{F}(Q^0_{\mathsf{A}},Q^1_{\mathsf{A}}) &\geq \left(1-\delta_{\mathsf{TD}}\left(Q^0_{\mathsf{A}},Q^1_{\mathsf{A}} \right) 
\right)^2\\
&\geq  \Bigg(1-\delta_{\mathsf{TD}}\left(\mathrm{Tr}_{\mathsf{B}}\left[Q^0_{\mathsf{AB}}\right], \hat{Q}^0_{\mathsf{A}} \right)-\delta_{\mathsf{TD}}\left( \hat{Q}^0_{\mathsf{A}},\mathrm{Tr}_{\mathsf{B}}\left[Q^1_{\mathsf{AB}}\right] \right)
\Bigg)^2\\
&\geq \left(1-2 \exp\left(-\frac{np}{24} \right)\right)^2 \,\,\,\geq \,\,\, 1-4 \exp\left(-\frac{np}{24} \right).
\end{align*}    
\end{proof}

Using \Cref{lem:fidelity-lemma}, we can now argue that $\mathsf{LSN}_{n,k,\mathcal{D}_p^{\otimes n}}$ admits unique solutions with overwhelming probability---provided that $n$ is slightly larger than $k$, say $n \geq 8k$, and that $p \in (0,1/2)$ is a sufficiently small constant, for example $p = 0.05$.

\begin{lemma}[Existence of unique solutions]\label{lem:existence-sol} Let $n,k \in \mathbb{N}$ be integers and let $p \in (0,1/2)$ be a parameter.
Then, the $\mathsf{LSN}_{n,k,\mathcal{D}_p^{\otimes n}}$ problem admits a unique solution and can be solved information-theoretically with probability at least $1-\delta$, where
$$ 
\delta \, \leq \, 3 np \cdot 2^{n H(3p) } \cdot 3^{3np} \cdot 2^{-n+k} + 2 \cdot e^{-\frac{np}{48}}. 
$$
\end{lemma}
\begin{proof}
Suppose we are given as input an instance in the Clifford representation, 
\begin{align}\label{eq:lsn-instance-cliff}
\big(C \in \Cliff_n, E_a \, C (\ket{0^{n-k}} \otimes \ket{ x})\big)  
\end{align}
where $C \sim  \Cliff_n$ is a random $n$-qubit Clifford, $E_a \sim \mathcal{D}_p^{\otimes n}$ is an $n$-qubit Pauli error from a local depolarizing channel, and $x \sim \bit^k$ is a random $k$-bit secret. 

Let $\ket{Q^0}_{\mathsf{AB}}$ and $\ket{Q^1}_{\mathsf{AB}}$
be the bipartite states in Eq.~\eqref{eq:Q0} and Eq.~\eqref{eq:Q1}, respectively, for some $p \in (0,1/2)$. Note that $\ket{Q^0}_{\mathsf{AB}}$ depends on the Clifford $C \in \Cliff_n$. Then, \Cref{lem:fidelity-lemma} shows that, with probability at least $ 1-3 np \cdot 2^{n H(3p) } \cdot 3^{3np} \cdot 2^{-n+k}$ over the choice of the random Clifford $C$, it holds that $\mathrm{F}(Q_{\mathsf{A}}^0,Q_{\mathsf{A}}^1) \geq 1-4 \cdot e^{-\frac{np}{24}}$.
Therefore, by Uhlmann's theorem, there exists a unitary $U \in \mathrm{L}(\algo H_{\mathsf{B}})$ which acts on the $\mathsf{B}$ register and maps $\ket{Q^0}_{\mathsf{AB}}$ to another state which achieves an overlap of at least $4\cdot e^{-\frac{np}{24}}$ with $\ket{Q^1}_{\mathsf{AB}}$. Applying the Uhlmann unitary $U$ on the quantum part of the input from \Cref{eq:lsn-instance-cliff} together with an additional ancilla register $\ket{0}$, i.e.,
$$
E_a C (\ket{0^{n-k}} \otimes \ket{ x}) \otimes \ket{0} \quad \mapsto \quad U \left( E_a \, C (\ket{0^{n-k}} \otimes \ket{ x}) \otimes \ket{0} \right)
$$
we obtain a state which---via Fuchs-van De Graaf---is within trace distance at most $2e^{-\frac{np}{48}}$ of the state $\ket{0^{n-k}} \otimes \ket{x} \otimes \ket{a}$.
In other words, the Uhlmann unitary $U$ allows us to solve $\mathsf{LSN}_{n,k,\mathcal{D}_p^{\otimes n}}$ and to recover the $x$ with the desired success probability.
\end{proof}

\subsection{Multi-Shot Variant}

Recall that the quantum learning problem $\mathsf{LSN}_{n,k,\mathcal{D}_p^{\otimes n}}$ in \Cref{def:LSN} only involves a single quantum sample of the form
$$
\big(S, E \ket{\overline{\psi_x}}^S \big) \,\,\sim \,\,\mathsf{LSN}_{n,k,\mathcal{D}_p^{\otimes n}}.
$$
One may reasonably ask: does the learning problem become easier if the learner instead receives many independently chosen samples $\{S_i, E_i \ket{\overline{\psi_x}}^{S_i}\}_{i \in [m]}$, where the secret $x$ remains the same throughout each sample? This motivates us to consider a multi-shot variant of the LSN problem, which we study in this section. Note that the multi-sample learning task also shows up naturally in the related $\lpn$ problem (see \Cref{sec:reduction-LPN}).
We define the multi-shot variant of the $\mathsf{LSN}_{n,k,\mathcal{D}_p^{\otimes n}}$ as follows.

\begin{definition}[Multi-Shot Learning Stabilizers with Noise]\label{prob:LSN_multi}
Let $k \in \mathbb{N}$ be the security parameter and let $n=\poly(k)$ be an integer. Let $p \in (0,1/2)$ be a parameter. The multi-shot Learning Stabilizers with Noise ($\mathsf{MSLSN}_{m,n,k,\mathcal{D}_p^{\otimes n}}$) problem is to find $x \in \bit^k$ given
$$
\big\{S_i\in \Stab(n,k), E_i \ket{\overline{\psi_x}}^{S_i} \big\}_{i=1}^m \, ,
$$
where, for each index $i \in [m]$, $S_i\sim \Stab(n,k)$ is a uniformly random stabilizer and $E_i \sim \mathcal{D}_p^{\otimes n}$ is a Pauli error $E \in \bar{\mathcal{P}}_n$, and where $x \sim \bit^k$ is a random string.
\end{definition}

Depending on the application, it may be more natural to think of $m$, the sample complexity, as a parameter that one wishes to minimize, rather than as a fixed parameter. As previously mentioned in \Cref{remark:clifford-rep}, we can equivalently use the Clifford representation and think of an $\mathsf{MSLSN}_{m,n,k,\mathcal{D}_p^{\otimes n}}$ instance a a tuple
$$
\left\{C_i \in \Cliff_n, \, E_i \, C_i (\ket{0^{n-k}} \otimes \ket{x})\right\}_{i=1}^m \, ,
$$
where $C_i \sim \Cliff_n$ and $E_i \sim \mathcal{D}_p^{\otimes n}$, for each index $i \in [m]$, and where $x \sim \bit^k$.  

It is not hard to see that $\mathsf{MSLSN}_{m,n,k,\mathcal{D}_p^{\otimes n}}$ reduces to (a special variant of) the single-shot problem in \Cref{def:LSN}. Indeed, we can show the following.

\begin{lemma}[$\mathsf{MSLSN}_{m,n,k,\mathcal{D}_p^{\otimes n}}$ reduces to a special variant of $\lsn$\label{lem:MSLSNtoLSN}]
Let $k \in \mathbb{N}$, let $n=\poly(k)$ be an integer and $p \in (0,1/2)$ be a parameter.
Then, there exist distributions  $\mathcal{S}$ and $\mathcal{I}$ which depend on $m,n,k$ and $p$ and have support over $\Cliff_{mn}$ and $\bit^{mk}$, respectively, such that
$\mathsf{MSLSN}_{m,n,k,\mathcal{D}_p^{\otimes n}}$ reduces to $\lsn_{mn,mk,\mathcal{D}_p^{\otimes mn},\mathcal{S},\mathcal{I}}$.
\end{lemma}
\begin{proof}
Suppose we are given as input an $\mathsf{MSLSN}_{m,n,k,\mathcal{D}_p^{\otimes n}}$ instance of the form
$$
\left\{C_i \in \Cliff_n, \, E_i \, C_i (\ket{0^{n-k}} \otimes \ket{x})\right\}_{i=1}^m \, ,
$$
where $C_i \sim \Cliff_n$ and $E_i \sim \mathcal{D}_p^{\otimes n}$, for each index $i \in [m]$, and where $x \sim \bit^k$. 

Consider the reduction which proceeds as follows:
\begin{enumerate}
    \item Let $\pi \in \frak{S}_{mn}$ be the permutation on $mn$ elements with permutation operator
    $$
Q(\pi) \left( (\ket{0^{n-k}} \otimes \ket{x}) \otimes \dots \otimes (\ket{0^{n-k}} \otimes \ket{x}) \right) = \ket{0^{m(n-k)}} \otimes \ket{x^{m}} \,
$$
where $x^m = x \cdots x$ such that $x$ is repeated $m$ times.

\item Run the solver for the $\lsn_{mn,mk,\mathcal{N},\mathcal{S},\mathcal{I}}$ problem on input
$$
\left(Q(\pi) \big(\bigotimes_{i \in [m]} C_i \big)Q(\pi)^\dag, \,Q(\pi) \bigotimes_{i \in [m]} E_i \, C_i (\ket{0^{n-k}} \otimes \ket{x}) \right).
$$
\end{enumerate}
Here, we define the corresponding distributions $\mathcal{N}$, $\mathcal{S}$ and $\mathcal{I}$ as follows:
\begin{itemize}
    \item $\mathcal{N}$ is the product distribution $\mathcal{D}_p^{\otimes mn}$ over local depolarizing channels.

    \item $\mathcal{S}$ is the distribution over Stabilizer codes which outputs Cliffords of the form
    $$
Q(\pi) \big(\bigotimes_{i \in [m]} C_i \big)Q(\pi)^\dag, \quad\quad \text{ for } C_i \sim \Cliff_n, \, \forall i \in [m].
$$

\item $\mathcal{I}$ is the distribution over $\bit^{mk}$ which first samples $x \sim \bit^k$, and then outputs $x^m = x \cdots x$ where $x$ is repeated $m$ times.
\end{itemize}
We now argue that our reduction allows us to solve $\mathsf{MSLSN}_{m,n,k,\mathcal{D}_p^{\otimes n}}$---provided that we have a solver for $\lsn_{mn,mk,\mathcal{N},\mathcal{S},\mathcal{I}}$. To this end, we observe that
\begin{align*}
&Q(\pi) \bigotimes_{i \in [m]} E_i \, C_i (\ket{0^{n-k}} \otimes \ket{x})\\
&= Q(\pi) \left(\bigotimes_{i \in [m]} E_i \, C_i \right) Q(\pi)^\dag Q(\pi)\left( (\ket{0^{n-k}} \otimes \ket{x}) \otimes \dots \otimes (\ket{0^{n-k}} \otimes \ket{x}) \right)\\
&=  Q(\pi) \left(\bigotimes_{i \in [m]} E_i \, C_i \right) Q(\pi)^\dag\left(\ket{0^{m(n-k)}} \otimes \ket{x^{m}}\right) \\
&=  Q(\pi) \left(\bigotimes_{i \in [m]} E_i \right) Q(\pi)^\dag 
Q(\pi)\left(\bigotimes_{i \in [m]}  C_i \right)Q(\pi)^\dag 
\left(\ket{0^{m(n-k)}} \otimes \ket{x^{m}}\right).
\end{align*}
Since the distribution $\mathcal{D}_p^{\otimes mn}$ is invariant under permutations via $Q(\pi)$, our reduction produces an instance of $\lsn_{mn,mk,\mathcal{N},\mathcal{S},\mathcal{I}}$, as desired.
\end{proof}

\section{Reduction from Learning Parity with Noise}\label{sec:reduction-LPN}

In this section, we show that LSN is a rich assumption which captures the classical $\lpn$ problem as a special case. Despite of more than a quarter of a century of study, the fastest
known (classical or quantum) algorithms for $\lpn$ still run in exponential time~\cite{BKW03}. Because the LSN problem (in some sense) subsumes $\lpn$, this can be seen as additional evidence for the average-case hardness of our learning problem.

In fact, our also results suggest that the LSN problem---precisely because it is an inherently quantum assumption---may in fact be even harder to break than its classical counterpart, which makes it particularly appealing as a basis of hardness in quantum cryptography.

\subsection{Learning Parity with Noise}

Recall that the $\lpn$ assumption~\cite{BFKL93} says that it is computationally difficult to decode a random linear code under Bernoulli noise. 

\begin{definition}[Learning Parity with Noise]\label{def:LPN}

Let $n,k \in \mathbb{N}$ and let $p \in (0,1/2)$ be a parameter. The Learning Parity with Noise ($\mathsf{LPN}_{n,k,\mathsf{Ber}_{p}^{\otimes n}}$) problem is to find $\vec x$ given as input
$$
(\vec A \sim \Z_2^{n\times k},\vec A \cdot \vec x +  \vec e \Mod{2}) 
$$
where $\vec x \sim \Z_2^{k}$ and $ \vec e \sim \Ber_p^{\otimes n}$ is a random Bernoulli error term.
\end{definition}

The $\lpn$ decoding problem is believed to be hard against both classical and quantum algorithms running in time $\poly(k)$ in the 
\emph{constant-noise} regime, where $p \in (0,1/2)$ is a constant and $n = \poly(k)$. In this regime, the celebrated BKW algorithm~\cite{BKW03} solves $\lpn$ with time/sample  complexity given by $O(2^{k/\log k})$.

Next, we study the relationship between $\lpn$ and LSN.

\subsection{Quantum Reduction to LSN}\label{subsec:reduction-LPN}

The goal of this section is make the connection between LSN and $\lpn$ more explicit; in particular, we show that $\lpn$ reduces to a special case of LSN. Specifically, we show that $\mathsf{LPN}_{n,k,\mathsf{Ber}_{p}^{\otimes n}}$ reduces to $\lsn_{n,k,\mathcal{N},\mathcal{S},\mathcal{I}}$ with respect to the following set of distributions:
\begin{itemize}
    \item $\mathcal{N}$ is the product $\mathcal{F}_{p}^{\otimes n}$ of bit-flip noise channels $\mathcal{F}_{p}$ with parameter $p >0$.

    \item $\mathcal{S}$ is the distribution over $\Cliff_n$ which samples $C \sim \algo S$ as follows: first, sample a random matrix $\vec A \in \Z_2^{n \times k}$ of full column-rank, and let $C$ be the corresponding matrix-multiplication Clifford operator $C: \ket{0^{n-k}} \otimes \ket{\vec x} \rightarrow \ket{\vec A \cdot\vec x \Mod{2}}$.

    \item $\mathcal{I}$ is the uniform distribution over $x \in \bit^k$.
\end{itemize}

Finally, we ask: can we also provide hardness results in the multi-shot setting? We give a quantum reduction from LPN which applies in both settings.

\paragraph{Single-shot variant.}

 First, we focus on the single-shot setting and show the following theorem relating LSN and $\lpn$.

\begin{theorem}[$\lpn$ reduces to LSN] Let $n,k \in \mathbb{N}$ be integers and let $p \in (0,1/2)$ be a parameter. 
Suppose there exists an algorithm $\mathcal{A}$ that runs in time $T$ and solves $\lsn_{n,k,\mathcal{N},\mathcal{S},\mathcal{I}}$ with probability $\epsilon$.
Then, there exists an algorithm $\mathcal{B}$ which runs in time $\poly(k,T)$
and solves
$\mathsf{LPN}_{n,k,\mathsf{Ber}_{p}^{\otimes n}}$ with probability at least $\epsilon \cdot \left( 1 - k \cdot 2^{k-n-1}\right).
$
\end{theorem}
\begin{proof}
Suppose we are given as input an instance
$$
(\vec A \in \Z_2^{n\times k},\vec A \cdot \vec s + \vec e \Mod{2}) \quad \sim \quad \mathsf{LPN}_{n,k,\mathsf{Ber}_{p}^{\otimes n}} 
$$
where $\vec A \sim \Z_2^{n\times k}$, $\vec s \sim \Z_2^{k}$ is a secret vector and $\vec e \sim \Ber_p^{\otimes n}$ is a random Bernoulli error term. Consider the reduction $\mathcal{B}$ which proceeds as follows:
\begin{itemize}
    \item If $\mathrm{col}\mbox{-}\mathrm{rank}(\vec A) < k$, $\mathcal{B}$ aborts.

    \item Else, if $\mathrm{col}\mbox{-}\mathrm{rank}(\vec A) = k$, $\mathcal{B}$ runs $\mathcal{A}$ on input
    $$
(U_{\vec A}, \ket{\vec A\cdot  \vec s + \vec e \Mod{2}})
$$
where $U_{\vec A}$ is the Clifford encoding circuit\footnote{Note that $U_{\vec A}$ can be described in terms of $\mathrm{CNOT}$ gates~\cite{10.5555/2011763.2011767}, which are themselves Clifford gates.} encoding operation
$$
U_{\vec A}: \ket{0^{n-k}} \otimes \ket{\vec x} \rightarrow \ket{\vec A \cdot\vec x \Mod{2}}
$$
which is an injective matrix multiplication for any vector $\vec x \in \Z_2^k$. In other words, $U_{\vec A}$ gives rise to the stabilizer subgroup $S_{\vec A} = \langle U_{\vec A}Z_1 U_{\vec A}^\dag, \dots, U_{\vec A}Z_{n-k} U_{\vec A}^\dag \rangle$.
\end{itemize}
Let is now analyze the probability that $\mathcal{B}$ succeeds. First, we observe that
$$
\underset{ \vec A \sim \Z_2^{n \times k}}{\Pr}[\mathrm{col}\mbox{-}\mathrm{rank}(\vec A) =k] = \prod_{i=1}^k \left(1 - 2^{i-n-1}\right) \geq \left(1 - 2^{k-n-1}\right)^k \geq 1 - k \cdot 2^{k-n-1}.
$$
Here, the last inequality follows from Bernoulli's inequality.
In other words, a uniformly random matrix $ \vec A \sim \Z_2^{n \times k}$ has full column-rank with overwhelming probability, provided that $n$ is only slightly larger than $k$.
We can interpret the noisy sample $\vec A\cdot \vec s + \vec e \Mod{2}$ 
as an ensemble of $n$-qubit pure states
\begin{align*}
\ket{\vec A\cdot \vec s + \vec e \Mod{2}}  &= X^{\vec e} \ket{\vec A\cdot  \vec s \Mod{2}}\\
&= X^{\vec e} U_{\vec A}\left(\ket{0^{n-k}} \otimes \ket{\vec s}\right) \, ,
\end{align*}
where $X^{\vec e} = X^{e_1} \otimes \dots \otimes X^{e_n}$ is a product of Pauli-X operators. Note that, since the error $\vec e \sim \Ber_p^{\otimes n}$ comes from a Bernoulli distribution, it follows that $X^{\vec e} \sim \mathcal{F}_p^{\otimes n}$ corresponds to an $n$-qubit bit-flip error with parameter $p$.

Therefore, we conclude that $\mathcal{B}$ runs in time $\poly(k,T)$ and solves
$\mathsf{LPN}_{n,k,\mathsf{Ber}_{p}^{\otimes n}}$ with probability at least $\epsilon \cdot \prod_{i=1}^k \left(1 - 2^{i-n-1}\right)$.
\end{proof}

\paragraph{Multi-shot variant.}

Finally, we consider the multi-shot setting and show the following theorem which relates $\lpn$ and MSLSN.

\begin{theorem}[$\lpn$ reduces to MSLSN] Let $m,n,k \in \mathbb{N}$ be integers and $p \in (0,1/2)$. 
Suppose there exists an algorithm $\mathcal{A}$ that runs in time $T$ and solves $\mathsf{MSLSN}_{m,n,k,\mathcal{N},\mathcal{S},\mathcal{I}}$ with probability $\epsilon$.
Then, there exists an algorithm $\mathcal{B}$ which runs in time $\poly(k,T)$
and solves
$\mathsf{LPN}_{nm,k,\mathsf{Ber}_{p}^{\otimes nm}}$ with probability at least $\epsilon  \cdot \left( 1 -  m\cdot k \cdot2^{k-n-1} \right)$.
\end{theorem}
\begin{proof}
Suppose we are given as input an instance
$$
(\vec A \sim \Z_2^{nm\times k} ,\vec A \cdot \vec s + \vec e \Mod{2}) \quad \sim \quad \mathsf{LPN}_{nm,k,\mathsf{Ber}_{p}^{\otimes nm}} 
$$
where $\vec A =[\vec A_1 | \cdots | \vec A_m]^\intercal \in \Z_2^{nm\times k}$ consists of $\vec A_i \sim \Z_2^{n\times k}$, $\vec s \sim \Z_2^{k}$ is a secret vector and $\vec e = [\vec e_1| \cdots | \vec e_m]^\intercal \sim \Ber_p^{\otimes nm}$ consists of random Bernoulli error terms $\vec e_i \sim \Ber_p^{\otimes n}$. Consider the reduction $\mathcal{B}$ which proceeds as follows:
\begin{itemize}
    \item If there exists an index $i \in [m]$ such that $\mathrm{col}\mbox{-}\mathrm{rank}(\vec A_i) < k$, $\mathcal{B}$ aborts.

    \item Else, if $\mathrm{col}\mbox{-}\mathrm{rank}(\vec A_i) = k$ for all $i \in [m]$, $\mathcal{B}$ runs $\mathcal{A}$ on input
    $$
\Big\{(U_{\vec A_i}, \ket{\vec A_i\cdot  \vec s + \vec e_i \Mod{2}})\Big\}_{i \in [m]}
$$
where $U_{\vec A_i}$ is the Clifford encoding circuit
$$
U_{\vec A_i}: \ket{0^{n-k}} \otimes \ket{\vec x} \rightarrow \ket{\vec A_i \cdot\vec x \Mod{2}}
$$
which is an injective matrix multiplication for any vector $\vec x \in \Z_2^k$.
\end{itemize}
Let is now analyze the probability that $\mathcal{B}$ succeeds. Again, we observe that
\begin{align*}
&\underset{ \vec A =[\vec A_1 | \cdots | \vec A_m]^\intercal\sim \Z_2^{nm \times k}}{\Pr}\Big[\forall i \in [m] \, : \,\mathrm{col}\mbox{-}\mathrm{rank}(\vec A_i) =k \Big]\\
&=\left(\prod_{i=1}^k \left(1 - 2^{i-n-1}\right)\right)^m
\geq \left(1 - 2^{k-n-1}\right)^{mk} \geq 1 - m \cdot k \cdot 2^{k-n-1}.
\end{align*}
Therefore, we conclude that $\mathcal{B}$ runs in time $\poly(k,m,T)$ and solves
$\mathsf{LPN}_{nm,k,\mathsf{Ber}_{p}^{\otimes nm}}$ with probability at least $\epsilon \cdot \left(1 - m \cdot k \cdot 2^{k-n-1}\right)
$.
\end{proof}

\section{Quantum Algorithms for Learning Stabilizers with Noise}
\label{sec:algorithms}

In this section, we give both
polynomial-time and exponential-time quantum algorithms for solving LSN in various depolarizing noise regimes.

\subsection{Single-Shot Decoding for Extremely Low Noise Rates}

In this section, we first consider an extremely low noise regime in which the LSN problem becomes computationally tractable. Specifically, we consider the $\mathsf{LSN}_{n,k,\mathcal{D}_{p}^{\otimes n}}$ problem with parameter $p=1/n$ and $n 
=\poly(k)$. We show that a simple projection onto the stabilizer code space suffices to solve $\mathsf{LSN}_{n,k,\mathcal{D}_{p}^{\otimes n}}$ in time $O(n^3)$ with good success probability. 

\begin{theorem}[Single-Shot Decoding for Extremely Low Noise Rates] Let $n,k \in \mathbb{N}$ be integers with $n \geq k$ and $n = \poly(k)$. Let $p =\frac{1}{n}$ be a noise parameter. Then, Algorithm $1$ runs in time $O(n^3)$ and solves the $\mathsf{LSN}_{n,k,\mathcal{D}_{p}^{\otimes n}}$ problem with probability at least $0.25$.
\end{theorem}
\begin{proof}
Suppose we are given as input an instance of the $\mathsf{LSN}_{n,k,\mathcal{D}}$ problem, i.e., $$
\big(S\in \Stab(n,k),\mathcal{D}_{p}^{\otimes n}(\proj{\overline{\psi_x}}^S) \big) \, \sim \mathsf{LSN}_{n,k,\mathcal{D}}$$
where $S \sim \Stab(n,k)$ describes a random stabilizer code and $x \sim \bit^k$ is a random element. Let $\rho_x^S \leftarrow \mathcal{D}_{p}^{\otimes n}(\proj{\overline{\psi_x}}^S)$ denote the ensemble instance.
Then, the probability that Algorithm $1$ succeeds is trivially bounded below by the probability that the depolarizing noise channel is error-less, and thus
\begin{align*}
\underset{x \sim \bit^k}{\mathbb{E}}\mathrm{Tr}\left[(\id_{n-k} \otimes \proj{x}) (U_{\mathrm{Enc}}^S)^\dag \rho_x^S U_{\mathrm{Enc}}^S\right] \geq \left(1-\frac{3p}{4}\right)^n
\geq 1 - n \cdot \frac{3p}{4} \, \geq \, 0.25 \, ,
\end{align*}
where the second inequality follows from Bernoulli's inequality, and the last inequality follows from our assumption that $p=1/n$.
\end{proof}

\begin{algorithm}[t]
\DontPrintSemicolon
\SetAlgoLined
\label{alg:projection-LSN}
\KwIn{Instance $(S\in \Stab(n,k),\mathcal{D}_{p}^{\otimes n}(\proj{\overline{\psi_x}}^S)) \sim \mathsf{LSN}_{n,k,\mathcal{D}_{p}^{\otimes n}}$.}
    
\KwOut{A string $x' \in \bit^k$.}

Let $\rho_x^S \leftarrow \mathcal{D}_{p}^{\otimes n}(\proj{\overline{\psi_x}}^S)$ denote the ensemble instance.

Use the algorithm in \Cref{thm:encode_stab} to find an encoding Clifford $U_{\mathrm{Enc}}^S$ associated with the stabilizer subgroup $S\in \Stab(n,k)$.

Compute $(U_{\mathrm{Enc}}^S)^\dag \rho_x^S U_{\mathrm{Enc}}^S$ and measure in the computational basis.

Output the string $x'$ corresponding to the last $k$ bits of the measurement outcome.

\caption{Projection onto the Stabilizer Codespace}
\end{algorithm}

\subsection{Single-Shot Decoding for Low Constant Noise Rates}

Recall that LSN is the following decoding problem: given as input an instance
$$
\big(S\in \Stab(n,k), E \ket{\overline{\psi_x}}^S \big) \, ,
$$
where $S\sim \Stab(n,k)$ is a random stabilizer code and $E \sim \mathcal{D}_p^{\otimes n}$ is a random $n$-qubit Pauli error $E \in \bar{\mathcal{P}}_n$, and where $x \sim [2^k]$ is a random element. In other words, we have a noisy state discrimination task $\{\rho_x\}_{x \in [2^k]}$ where $\rho_x$ is the mixed state 
\begin{equation}\label{eq:rhox}
\rho_x = \mathcal{D}_{p}^{\otimes n}(\proj{\overline{\psi_x}})= \sum_{E \in \bar{\mathcal{P}}_n} \Pr_{E \sim \mathcal{D}_{p}^{\otimes n}}[E] \cdot E \proj{\overline{\psi_x}}  E^\dag.
\end{equation}
Our algorithm for solving LSN is to implement a \emph{Pretty Good Measurement} (PGM) \cite{barnum2000reversingquantumdynamicsnearoptimal,Montanaro_2007}, but with a twist that enables us to bound its success probability: we will implement an approximation of the PGM that works for {\em truncated depolarizing noise}, noise from which we have culled the highest weight (and thus most destructive) errors. The reason this works is that the PGM is actually the optimal measurement for orthogonal states and discriminates them perfectly. Using the Gilbert-Varshamov bound, we are able to harness the fact that under truncated depolarizing noise, encoded orthogonal states remain approximately orthogonal, and thus the PGM still works well for the task.\footnote{This can be understood in another way: the purpose of a good error-correcting code is to encode quantum information in subspaces that do not contract much under quantum channels (i.e. two initial orthogonal states remain approximately orthogonal under the action of quantum channels)---a random stabilizer code is an example of such a code.}  


\paragraph{Pretty Good Measurement.} We recall the following result by Montanaro: 
\begin{lemma}[\cite{Montanaro_2007}]\label{lem:PGM-analysis}
Let $\{\rho_x\}$ be an ensemble of quantum states and let $\boldsymbol{\Lambda}_x = \boldsymbol{\Sigma}^{-\frac{1}{2}} \rho_x \boldsymbol{\Sigma}^{-\frac{1}{2}}$ with $\boldsymbol{\Sigma} = \sum_x \rho_x$,  where inverses of $\boldsymbol{\Sigma}$ are taken with respect to its support. Then, the worst-case error of the pretty good measurement ensemble $\{\boldsymbol{\Lambda}_x\}$ is at most
$$
p_{\text{err}} = \max_x \left(1 - \Tr{\boldsymbol{\Lambda}_x \rho_x} \right) \, \leq \, \sum_{x \neq y} \sqrt{\mathsf{F}(\rho_x,\rho_y)}
$$
where $\mathsf{F}(\rho_x,\rho_y)$ denotes the fidelity between $\rho_x$ and $\rho_y$. Moreover, the PGM is optimal if the states in $\{\rho_x\}$ are pair-wise orthogonal, as then $\mathsf{F}(\rho_x,\rho_y) = 0.$
\end{lemma}

We're going to use the following block-encoding based algorithm in~\cite{PhysRevLett.128.220502} for implementing the PGM.
Let $\kappa_{\rho}$ denote the reciprocal of the smallest eigenvalue of a density matrix $\rho$. We use the following theorem.

\begin{theorem}[\cite{PhysRevLett.128.220502}]\label{thm:PGM} The PGM measurement channel for $\{\rho_x\}_{x \in  [2^k]}$ can be implemented with error $\epsilon$ (in terms of diamond distance) in time
$$
\widetilde{\mathcal{O}}\left(\sqrt{2^k \kappa_{\bar{\rho}}} N_{\rho} (\kappa_{\bar{\rho}}  +  \min(\kappa_\rho,2^k \cdot \kappa_{\bar{\rho}}/\epsilon^2)\right),
$$
where $\bar{\rho} = 2^{-k}\sum_{x \in [2^k]} \rho_x$ and where $N_{\rho}$ denotes the size of the quantum circuit needed to implement a purification of $\rho_x$.
\end{theorem}

In our case, since a purification of $\rho_x$ is 
\begin{equation}
    \sum_{E \in \bar{\mathcal{P}}_n} \sqrt{\Pr_{E_a \sim \mathcal{D}_{p}^{\otimes n}}[E_a]} \ket{a} \otimes E_a \UES (\ket{0^{n-k}}\otimes \ket{x}),
\end{equation}
we may prepare this purification simply by applying $\UES$ on a coherent superposition. The number of gates required to implement an $n$-qubit Clifford is $N_{\rho} = O(n^2)$.
We show the following theorem.

\begin{theorem}[Single-Shot Decoding for Low Constant Noise Rates\label{thm:singleshotdecoding}] Let $n,k \in \mathbb{N}$ and $\epsilon \in (0,1)$. Let $\mathcal{D}_p^{\otimes n}$ be the $n$-qubit depolarizing channel, for some $p\in (0,1/2)$ such that
\begin{equation}
    H(3p)+3\log_2(3) p + k/n< .99 -\frac{\log(1/\epsilon)}{n}.
\end{equation}
Then, there exists a quantum algorithm which runs in time $$
\widetilde{\mathcal{O}}\left(n^2\sqrt{2^k \kappa_{\bar{\rho}}}(\kappa_{\bar{\rho}}  +  \min(\kappa_\rho,2^k \cdot \kappa_{\bar{\rho}}/\epsilon^2)\right).
$$
and solves the $\mathsf{LSN}_{n,k,\mathcal{D}_p^{\otimes n}}$ problem with probability at least $1 - O(\epsilon).$
\end{theorem}
\begin{proof}
Suppose we are given as input an instance of the $\mathsf{LSN}_{n,k,\mathcal{D}_p^{\otimes n}}$ problem, i.e., $$
\big(S\in \Stab(n,k),\mathcal{D}_{p}^{\otimes n}(\proj{\overline{\psi_x}}^S) \big) \, \sim \mathsf{LSN}_{n,k,\mathcal{D}_p^{\otimes n}}$$
where $S \sim \Stab(n,k)$ describes a random stabilizer code and $x \sim \bit^k$ is a random element.
We now show that running Algorithm $2$ on input $(S,\mathcal{D}_{p}^{\otimes n}(\proj{\overline{\psi_x}}^S))$ and parameter $\epsilon \in (0,1)$ yields $x$ with the desired success probability.

Algorithm $2$, in fact, implements the PGM with respect to a slightly different ensemble as compared to the true ensemble of problem instances. That is to say, it will suffice to implement the pretty good measurement with respect to the ensemble $\{\tilde{\rho}_x^S\}_{x\in \{0,1\}^k}$ where
\begin{equation}\label{eq:truncated_states_PGM}
\tilde{\rho}_x^S := \tilde{\mathcal{D}}^{(\frac{3np}{2})}(\proj{\overline{\psi_x}}^S)= \sum_{E \in \bar{\mathcal{P}}_n} \tilde{\vec p}^{(\frac{3n}{2})}(E) \cdot E \proj{\overline{\psi_x}}^S  E^\dag,
\end{equation}
where $\tilde{\mathcal{D}}^{(\frac{3np}{2})}$ is the truncated depolarizing noise channel to be defined shortly; while we remind readers that the true ensemble of problem instances is 
\begin{equation}\label{eq:nontruncated_states_PGM}
\rho_x^S := \tilde{\mathcal{D}}^{(\frac{3np}{2})}(\proj{\overline{\psi_x}}^S)= \sum_{E \in \bar{\mathcal{P}}_n} \tilde{\vec p}^{(n)}(E) \cdot E \proj{\overline{\psi_x}}^S  E^\dag.
\end{equation}

We now define the {\em truncated depolarizing noise channel} as the channel which acts on an input state $\rho$ as
\begin{equation}
\tilde{\mathcal{D}}^{(w)}(\rho) := \sum_{E\in \mathcal{P}_n: |E|\leq w} \tilde{\vec p}^{(w)}(E) E\rho E^{\dagger}
\end{equation}
where $p \in (0,3/4)$ denotes the depolarizing noise parameter (c.f. \Cref{eq:nqubitdep}), and the truncation lies in the fact that we restrict the support to Paulis with bounded weight. We define the truncated probability distribution via
\begin{equation}
    \tilde{\vec p}^{(w)}(E) := \frac{1}{N}\left(\frac{p}{3}\right)^{|E|} \left(1-p\right)^{n-|E|}
\end{equation}
and $N$ is a normalization factor, i.e. $N = \sum_{E\in \bar{\calP}_n: |E|\leq w} \tilde{\vec p}^{(w)}(E),$ that ensures that $\tilde{\vec p}^{(w)}$ is a probability distribution over $\bar{\calP}_n$. The distribution corresponding to the $n$-qubit local depolarizing noise channel corresponds to $w=n,$ i.e. $\tilde{\vec p}^{(n)}(E) = \Pr_{E \sim \mathcal{D}_{p}^{\otimes n}}[E]$ and the truncated channel corresponds to acting only with the Paulis with weight at most $w$, with the same relative probabilities as in $n$-qubit local depolarizing noise. Because the weights of the Pauli channels in the decomposition of $\mathcal{D}_p^{\otimes n}$ are distributed as a binomial $w \sim \mathsf{Binom}(n,p)$, one can show that for $w=3np/2$, the total variation distance between the probability distribution over the $n$-qubit Pauli channels induced by $\mathcal{D}_p^{\otimes n}$ and $\tilde{\mathcal{D}}^{(w)}$ is
\begin{equation}\label{eq:TV}
\delta_{\mathsf{TV}}\left(\tilde{\vec p}^{(3np/2)}, \tilde{\vec p}^{(n)}\right) \leq \Pr_{E \sim \mathcal{D}_{p}^{\otimes n} }\left[|E| \geq \frac{3}{2}np\right] \leq \exp\left(-\frac{np}{12} \right),
\end{equation}
using a Chernoff bound.
    
\begin{algorithm}[t]
\DontPrintSemicolon
\SetAlgoLined
\label{alg:PGM-LSN}
\KwIn{Instance $(S\in \Stab(n,k),\mathcal{D}_{p}^{\otimes n}(\proj{\overline{\psi_x}}^S)) \sim \mathsf{LSN}_{n,k,\mathcal{D}_{p}^{\otimes n}}$ and $\epsilon \in (0,1)$.}
    
\KwOut{A string $x' \in \bit^k$.}

Let $\rho_x^S \leftarrow \mathcal{D}_{p}^{\otimes n}(\proj{\overline{\psi_x}}^S)$ denote the ensemble instance.

Use the algorithm in \Cref{thm:PGM} with precision $\epsilon$ to measure $\rho_x^S$ via the $\mathsf{POVM}$
$$
\{\tilde{\boldsymbol{\Lambda}}_x^S\}_{x \in \bit^k} \quad \text{ with } \quad
\tilde{\boldsymbol{\Lambda}}_x^S = \boldsymbol{\Sigma}^{-\frac{1}{2}} \tilde{\rho}_x^S \, \boldsymbol{\Sigma}^{-\frac{1}{2}}
$$
where $\boldsymbol{\Sigma} = \sum_x \tilde{\rho}_x^S$ and where inverses of $\boldsymbol{\Sigma}$ are taken with respect to its support, and where the state $\tilde{\rho}_x^S$ is defined in \Cref{eq:truncated_states_PGM}.

Output the measurement outcome $x' \in \bit^k$.

\caption{Pretty Good Measurement for $\mathsf{LSN}$}
\end{algorithm}

Algorithm $2$ implements the PGM measurement channel with (diamond distance) error $\epsilon \in (0,1)$ with the stated time complexity. Let us first analyze the error probability of the (ideal) pretty good measurement $\{\tilde{\boldsymbol{\Lambda}}_x^S\}_{x}$:
\begin{align}
    \max_{x} \left(1 - \operatorname{tr}(\widetilde{\boldsymbol{\Lambda}}_x^S \rho_x^S)\right)  &= \max_x\left(1-\operatorname{tr}\left(\widetilde{\boldsymbol{\Lambda}}_x\left(\rho_x^S-\tilde{\rho}_x^S\right)\right)-\operatorname{tr}\left(\widetilde{\boldsymbol{\Lambda}}_x \tilde{\rho}_x^S\right)\right)\\
    &\leq \max_x \left(1+ \delta_{\mathsf{TD}}\left(\rho_x^S,\tilde{\rho}_x^S\right)-\operatorname{tr}\left(\widetilde{\boldsymbol{\Lambda}}_x \tilde{\rho}_x^S\right)\right)\\
    &\leq \delta_{\mathsf{TV}}\left(\tilde{\vec p}^{(n)}, \tilde{\vec p}^{(w)}\right) + \max_x \left(1-\operatorname{tr}\left(\widetilde{\boldsymbol{\Lambda}}_x \tilde{\rho}_x^S\right)\right)\\
    &=e^{-\frac{np}{12}} + \max_x \left(1-\operatorname{tr}\left(\widetilde{\boldsymbol{\Lambda}}_x \tilde{\rho}_x^S\right)\right)\\
    &\leq e^{-\frac{np}{12}} + \sum_{x \neq y} \sqrt{\mathsf{F}(\tilde{\mathcal{D}}^{(\frac{3np}{2})}(\proj{\overline{\psi_x}}^S),\tilde{\mathcal{D}}^{(\frac{3np}{2})}(\proj{\overline{\psi_y}}^S))}.\label{eq:fidelities}
\end{align}
The second inequality comes from the strong convexity of trace distance and \Cref{eq:truncated_states_PGM,eq:nontruncated_states_PGM}. The second last equality comes from Eq. \ref{eq:TV}. The last inequality comes from Lemma \ref{lem:PGM-analysis}.

To finish the proof, it suffices to bound the pair-wise fidelities in \Cref{eq:fidelities}. Here, we exploit the special structure of the encoded states. We appeal to the Gilbert-Varshamov bound (Theorem \ref{thm:GVbound}), which states that random stabilizer codes are non-degenerate (i.e. have good distance) with high probability. This means that errors of weight at most $\frac{3}{2}np$ acting on orthogonal states keep them orthogonal. Concretely, this implies that
\begin{align}
&\underset{S \sim \Stab(n,k)}{\Pr}\left[ E_a^\dag E_b \notin N(S), \, \forall \, |E_a|,|E_b| \leq \frac{3}{2}np \right] \\
&= \underset{S \sim \Stab(n,k)}{\Pr}\left[ \bra{\overline{\psi_x}}^S E_a^\dag E_b \ket{\overline{\psi_y}}^S = 0 \,\forall x\neq y, \forall |E_a|,|E_b| \leq \frac{3}{2}np \right] \\
&\geq 1- 3/2 np \cdot 2^{n H(3/2p) } \cdot 3^{3/2np} \cdot 2^{-n+k}.
\end{align}
The first equality follows from the observations mentioned above (Theorem \ref{thm:GVbound}). Moreover, it follows from Uhlmann's theorem that
\begin{align}
 &\mathsf{F}(\tilde{\mathcal{D}}^{(\frac{3np}{2})}(\proj{\overline{\psi_x}}^S),\tilde{\mathcal{D}}^{(\frac{3np}{2})}(\proj{\overline{\psi_y}}^S))  \\
 &= \max_U \left|\bra{\phi^{\rho_x}} (U \otimes I) \ket{\phi^{\rho_y}}\right|^2 \\
 &= \max_U \left| \sum_{E_a, E_b\in \bar{\calP}_n} \sqrt{\tilde{\vec p}^{(\frac{3np}{2})}(E_a)\tilde{\vec p}^{(\frac{3np}{2})}(E_b)} \bra{a}U\ket{b} \cdot \bra{\overline{\psi_x}} E_a^\dag E_b \ket{\overline{\psi_y}} \right| ^2 \label{eq:2}
\end{align}
where, for $x \in \bit^k$, we defined the purification
$$
\ket{\phi^{\rho_x}} = \sum_{E_a\in \bar{\calP}_n} \sqrt{\tilde{\vec p}^{(\frac{3np}{2})}(E_a)} \, \ket{a} \otimes E_a \ket{\overline{\psi_x}}^S.
$$
Note that the superposition ranges over Pauli errors of weight at most $\frac{3}{2}np$ as we are using the truncated distribution over Paulis. Therefore, conditioned on the event that the stabilizer code $S$ is non-degenerate and has distance at least $d = 3np+1$, the Knill-Laflamme error correction conditions imply that $\bra{\overline{\psi_x}} E_a^\dag E_b \ket{\overline{\psi_y}}=0$, 
for any pair of codewords with $x,y \in \bit^k$, and thus
\begin{equation}
    \mathsf{F}(\tilde{\mathcal{D}}^{(\frac{3np}{2})}(\proj{\overline{\psi_x}}^S),\tilde{\mathcal{D}}^{(\frac{3np}{2})}(\proj{\overline{\psi_y}}^S)) =0.
\end{equation} 
Further conditioning on the event that the implementation of the PGM succeeded, the probability of error due to the PGM mis-identifying the state (\Cref{eq:tobound}) is $e^{-{np}/12}$.

We can now put everything together to compute the final success probability of Algorithm $2$, union bounding over all the three error sources:
\begin{enumerate}
    \item $S$ is degenerate.
    \item Diamond-distance approximation error of implementing the PGM channel.
    \item PGM measurement mis-identifies $x$.
\end{enumerate}
We get that Algorithm $2$ successfully outputs the correct $x$ with probability at least $1-\delta$, where 
\begin{equation}
\delta \,\leq \, 3 np \cdot 2^{n H(3p) } \cdot 3^{3np} \cdot 2^{-n+k} + \epsilon + e^{-np/12}.
\end{equation}

This algorithm thus succeeds with constant probability for any noise rate $p$ such that the term 
\begin{equation}
    3 np \cdot 2^{n H(3p) } \cdot 3^{3np} \cdot 2^{-n+k} = \text{exp}[\log(3np)+ n (H(3p) + 3\log(3) p + k/n - 1)] 
\end{equation}
does not blow up. Noting that $k/n = O(1)$ in the definition of $\lsn$, we can then check that as long as $p$ is a constant that satisfies
\begin{equation}
    H(3p)+3\log_2(3) p + k/n< .99 -\frac{\log(1/\epsilon)}{n},
\end{equation}
this term vanishes exponentially in $n$ and the total probability of error is $O(\epsilon)$.
\end{proof}
\subsection{Multi-Shot Decoding Up to a Threshold}

By taking more samples, we can slightly increase the noise rate at which decoding is still possible, but this only works up to a certain noise threshold that we also compute. 

The algorithm is to run a modified PGM on a larger state space defined by the tensor product of the state spaces of all the samples. Note that we cannot appeal to \Cref{lem:MSLSNtoLSN} to argue that the algorithm and proof in the previous section carries over with the trivial replacement $n\leftarrow mn$, because the $\lsn$ problem on the larger state space corresponding to the $\mslsn$ problem does not have the same distribution over stabilizers as in the canonical $\lsn$. 


Let the support of every sample define a {\em block}. For $P\in \calP_{mn}$, let the {\em block support} of $P$ (denoted $\mathsf{BSupp}(P)$) be the number of blocks on which $P$ has at least one non-identity Pauli. We will call an error $E\in \calP_{mn}$ {\em typical} if it has full block support (i.e. $\mathsf{BSupp}(E) = m$) and on each block its weight is between $[1,3/2np]$. Let us call the set of all typical errors $\calP_{typical}$. 
\begin{theorem}[Multi-shot decoding at higher noise\label{thm:multishotdecoding}]Let $n,k \in \mathbb{N}$ and $\epsilon \in (0,1)$. Let $\mathcal{D}_p^{\otimes n}$ be the $n$-qubit depolarizing channel, for some $p\in (0,1/2)$ such that
\begin{equation}
    H(3p)+3\log_2(3) p + k/n< .99 -\frac{\log(1/\epsilon)}{mn}.
\end{equation} 
Then, there exists a quantum algorithm which runs in time 
$$
 \widetilde{\mathcal{O}}\left(n^2 m^2\sqrt{2^{mk}\kappa_{\bar{\rho}}}(\kappa_{\bar{\rho}}  +  \min(\kappa_\rho,2^{mk} \cdot \kappa_{\bar{\rho}}/\epsilon^2)\right)
 $$
and solves the $\mathsf{MSLSN}_{n,k,m,\mathcal{D}_p^{\otimes n}}$ problem (equivalently the $\lsn_{n,k,\mathcal{D}_p^{\otimes n}}$ problem with $m$ samples) with probability at least $1 - O(\epsilon).$
\end{theorem}

\begin{proof}
 Suppose we are given as input an instance of the $\mathsf{MSLSN}_{n,k,\mathcal{D}_p^{\otimes n},m}$ problem, i.e., $$
\big(\otimes_{i=1}^m S_i\in \Stab(n,k)^{\otimes m},\mathcal{D}_{p}^{\otimes mn}(\otimes_{i=1}^m \proj{\overline{\psi_x}}^{S_i}) \big) $$ where each $S_i \sim \Stab(n,k)$ describes a random stabilizer code and $x \sim [2^k]$ is a random element. 

For a given MSLSN instance, define $T=\otimes_{i=1}^m S_i \in \Stab(mn,mk)$. Let us denote the true ensemble instance by
\begin{equation}
     \rho_x^T \leftarrow \mathcal{D}_{p}^{\otimes mn}(\otimes_{i=1}^m \proj{\overline{\psi_x}}^{S_i}). 
\end{equation}
We will run the PGM $\{(\tilde{\mathbf{\Lambda}}_x^T)\}_{x\in \{0,1\}^k}$ which is the PGM relative to the modified ensemble $\{(\tilde{\rho}_x^T)\}_{x\in \{0,1\}^k}$ where
\begin{equation}\label{eq:truncated_states_PGM_2}
\tilde{\rho}_x^T := \mathcal{D}^{typical}(\otimes_{i=1}^m \proj{\overline{\psi_x}}^{S_i})
\end{equation}

That is, we approximate the true noise channel $\mathcal{D}_{p}^{\otimes mn}$ by a channel $\mathcal{D}^{typical}$ that still has tensor product structure across the blocks:
\begin{equation}
\mathcal{D}^{typical}(\rho) := \sum_{P\in \mathcal{P}_{typical}} p^{typical}(P) P\rho P^{\dagger} = \tilde{\mathcal{D}}^{(3/2np)\otimes m}(\rho).
\end{equation}
This is precisely the channel on the Hilbert space of $nm$ qubits where the truncated noise channel for the single-shot problem, $\tilde{\mathcal{D}}^{(3/2np)}$, is applied on each $n$-qubit block. 
As before, this is a good approximation of the distribution over the corresponding Pauli channels: 
\begin{equation}\label{eq:TV_2}
\left\lVert p^{typical} - \tilde{p}^{(mn)}\right\rVert_1 \leq m \left\lVert\tilde{p}^{(3/2np)} - \tilde{p}^{(n)}\right\rVert_1 \leq m\exp\left(-\frac{np}{12} \right),
\end{equation}
where the first inequality follows from the fact that both $\tilde{p}^{(mn)}$ and $p^{typical}$ are product distributions, so we may use subadditivity of total variation distance over product distributions. 

Next, we explicitly analyze the error probability of the (ideal) pretty good measurement given by $\{(\tilde{\mathbf{\Lambda}}_x^T)\}_{x\in \{0,1\}^k}$. Specifically, we find that
\begin{align}
    &\max_{x\in \{0,1\}^k} 1 - \operatorname{tr}(\tilde{\mathbf{\Lambda}}_x^T \rho_x^T) \\
    &\leq me^{-\frac{np}{12}} + \max_x \left(1-\operatorname{tr}\left(\widetilde{\boldsymbol{\Lambda}}_x ^T\tilde{\rho}_x^T\right)\right)\\
    &\leq me^{-\frac{np}{12}} + \sum_{x \neq y} \sqrt{\mathsf{F}(\mathcal{D}^{typical}(\otimes_{i=1}^m \proj{\overline{\psi_x}}^{S_i}),\mathcal{D}^{typical}(\otimes_{i=1}^m \proj{\overline{\psi_y}}^{S_i})}\\
    &\leq me^{-\frac{np}{12}} + \sum_{x \neq y} \sqrt{\prod_{i=1}^m \mathsf{F}(\tilde{\mathcal{D}}^{(\frac{3np}{2})}(\proj{\overline{\psi_x}}^{S_i}),\tilde{\mathcal{D}}^{(\frac{3np}{2})}( \proj{\overline{\psi_y}}^{S_i})}\label{eq:tobound}
\end{align}
where we have used the fact that fidelity is multiplicative across tensor product: $$F\left(\rho_1 \otimes \rho_2, \sigma_1 \otimes \sigma_2\right)=F\left(\rho_1, \sigma_1\right) F\left(\rho_2, \sigma_2\right).$$
Call $S_i$ ``good" if has the property that 
\begin{equation}
    \mathsf{F}(\tilde{\mathcal{D}}^{(\frac{3np}{2})}(\proj{\overline{\psi_x}}^{S_i}),\tilde{\mathcal{D}}^{(\frac{3np}{2})}( \proj{\overline{\psi_y}}^{S_i})=0, \quad \forall x,y.
\end{equation}
Note that as long as at least one $S_i$ is ``good", the entire term
\begin{equation}
\sum_{x \neq y} \sqrt{\prod_{i=1}^m \mathsf{F}(\tilde{\mathcal{D}}^{(\frac{3np}{2})}(\proj{\overline{\psi_x}}^{S_i}),\tilde{\mathcal{D}}^{(\frac{3np}{2})}( \proj{\overline{\psi_y}}^{S_i})}
\end{equation}
vanishes, and 
\begin{equation}
    \max_{x\in \{0,1\}^k} 1 - \operatorname{tr}(\tilde{\mathbf{\Lambda}}_x^T \rho_x^T)  \leq me^{-\frac{np}{12}}.
\end{equation}
This fortuitous event happens with probability 
\begin{equation}
    1- P(\text{all $S_i$ are bad}) = 1- (3/2 np \cdot 2^{n H(3/2p) } \cdot 3^{3/2np} \cdot 2^{-n+k})^m.
\end{equation}
Finally, accounting for the same three error sources as in the proof of Theorem \ref{thm:PGM}, the total probability of failure of the whole algorithm is
\begin{equation}
\delta \,\leq \, \left(\frac{3np \cdot 2^{n H(3p) } \cdot 3^{3np}}{2^{n-k}}\right)^m + \epsilon + me^{-np/12}.
\end{equation}
So as long as 
\begin{equation}
    H(3p)+3\log_2(3) p + k/n< .99 -\frac{\log(1/\epsilon)}{mn},
\end{equation}
and $m$ is polynomial in $n$, the total probability of error is $O(\epsilon)$.
\end{proof}
While the above analysis shows that asking for more samples in $\mslsn$ will increase the noise rate $p$ at which the PGM still recovers the secret, there is a threshold level of noise at which taking more samples can never help:
\begin{remark}[Noise threshold for PGM]
Decoding, or error recovery, becomes even information-theoretically intractable at some threshold noise rate. Our PGM-based algorithms fail at noise rate $p$ whenever
\begin{equation}
    H(3p)+3\log_2(3) p + k/n > 1.
\end{equation}
\end{remark}

\section{Worst-Case to Average-Case Reductions}\label{sec:worst-avg-reduction}

In this section, we provide further evidence for the hardness of $\lsn$ by showing a worst-to-average-case reduction for a variant of the problem. Let us intuitively explain the importance of this. By definition, the worst-case instance within a problem class is harder to solve than all other instances. However, a worst-to-average-case reduction states that an algorithm that succeeds with high probability over a uniformly chosen instance within that class (i.e., succeeds in the average case) would also suffice to solve the worst-case instance within that class. The upshot is that ``most" instances within that class are hard. 

In complexity theory, such reductions are regarded as critical pieces of evidence that a problem is indeed as hard as conjectured. While a worst-to-average-case reduction for $\mathsf{LWE}$ was already identified in \cite{Regev05}, $\lpn$ resisted similar attempts until the work of Brakerski, Lyubashevsky, Vaikuntanathan and Wichs~\cite{BLVW} in 2018. 


Our quantum reduction for $\lsn$ proceeds by entirely different means and uses a unitary ``twirl" to randomize the secret, the code and the error all at once. 
Unfortunately, our reduction only applies to a variant of LSN where the average-case instance has some mild dependence on the worst-case instance. We leave it as an open question whether it is possible to reduce to a much larger class of average-case problems that is completely independent of the worst-case instance.

\subsection{Overview of the Reduction} Suppose we are presented with a worst-case LSN instance of the form
\begin{equation}
    (S, E\ket{\overline{\psi_x}}^S )
\end{equation}
where $S\in \text{Stab}(n,k)$ is a stabilizer subgroup, $E\in \bar{P}_n$ is a Pauli error and $x\in \{0,1\}^k$ is a hidden secret---each potentially chosen adversarially. The goal of this section is to turn such a worst-case instance into an average-case instance of the LSN problem. Specifically, we will show how to draw a \emph{re-randomizing} Clifford unitary $R \in \Cliff_n$, that simultaneously re-randomizes
\begin{itemize}
    \item the underlying secret $x \in \bit^k$ of the instance, as well as

    \item the error $E$ and the underlying stabilizer subgroup $S$ of the instance. 
\end{itemize}
In other words, our reduction applies $R$ to the quantum part of the input, thereby obtaining a new (and ideally re-randomized) state of the form
\begin{equation}
    RE\ket{\overline{\psi_x}}^S = E' \ket{\overline{\psi_{x'}}}^{S'}.
\end{equation}
In the next sections, we show how to perform these steps separately.
First, we show in \Cref{lem:rerandomize_secret} how to re-randomize the LSN secret. Next, in \Cref{lem:rerandomize_codeerror}, we show how to re-randomize both the error and the underlying code. Finally, in \Cref{thm:Reduction_LSN} we put everything together and obtain the desired worst-case to average-case reduction.

\subsection{Re-Randomization of the Secret}

We now show how to sample a Pauli operator $P\in \calP_n$ that allows us to re-randomize the secret which underlies the LSN sample. 

For any stabilizer $S \in \Stab(n,k)$, we let $\algo L_X(S)= \{\overline{X^{u}}\}_{u\in \{0,1\}^k}$ denote the set of the logical Pauli $X$ operators associated with $S$. 
While the choice of logical Paulis associated with a given stabilizer code $S$ is not unique, we will use the prescription
\begin{equation}\label{eq:logicalPaulis}
    \overline{P} = U_{\text{Enc}}^S P (U_{\text{Enc}}^S)^{\dagger}, \quad \text{ for } P \in \algo P_n.
\end{equation}
We show that it suffices to apply a random Pauli in $\algo L_X(S)$ to re-randomize the secret.  

\begin{lemma}[Re-randomization of secret\label{lem:rerandomize_secret}]
Suppose that
$(S, E\ket{\overline{\psi_x}}^S)$ is a fixed instance, for some
stabilizer $S\in \Stab(n,k)$, error $E\in \bar{P}_n$ and secret $x\in \{0,1\}^k$.
Let $u \sim \bit^k$ be a random string and
let $\overline{X^{u}}\in \algo L_X(S)$ denote the logical Pauli $X^u$ with respect to $S$.
Then,
    \begin{equation}
        \overline{X^{u}} E\ket{\overline{\psi_x}}^S = E\ket{\overline{\psi_{x \oplus u}}}^{S}.
    \end{equation}
Moreover, the distribution of $x \oplus u$ is now uniform over $\bit^k$.
\end{lemma}

\begin{proof}
First, we observe that     $\overline{X^{u}}E\ket{\overline{\psi_x}}^S = \pm E\overline{X^{u}}\ket{\overline{\psi_x}}^S $ because Paulis either commute or anticommute. We adopt the convention of ignoring global phases, so we will evaluate $E\overline{X^{u}}\ket{\overline{\psi_x}}^S$ from now on. Calculating the action of $\overline{X^{u}}$ on the state, we find
      \begin{align}
    EP\ket{\overline{\psi_x}}^S &= E\overline{X^{u}}U_{\text{Enc}}^S(\ket{0^{n-k}}\otimes\ket{x}) \\
       &= EU_{\text{Enc}}^S\left((\id_{n-k}\otimes X^{u})\ket{0^{n-k}}\otimes\ket{x} \right) \\
       &=  EU_{\text{Enc}}^S(\ket{0^{ n-k}}\otimes\ket{x\oplus u}).
    \end{align}
   Because $u\sim \{0,1\}^k$ is random, it follows that $x\oplus u$ is also uniformly distributed for any fixed $x \in \bit^k$. This proves the claim.
\end{proof}

\subsection{Re-Randomization of the Code and the Error}

We now show how to re-randomize both the stabilizer $S$ and the error $E$ of a particular LSN instance. There is an important subtlety: the two cannot be randomized independently of each other. Acting with some Clifford unitary $U$ on a given noisy codeword $E\ket{\overline{\psi_x}}^S$ re-randomizes the code $S$ and the error $E$ simultaneously, via:
\begin{equation}\label{eq:correlatedcodeerror}
    U\left(E \ket{\overline{\psi_{x}}}^S\right) = (UEU^{\dag})U\ket{\overline{\psi_{x}}}^S = (UEU^{\dag}) \ket{\overline{\psi_{x}}}^{USU^\dag}.
\end{equation}
It is clear from \Cref{eq:correlatedcodeerror} that the new code (with stabilizer $USU^{\dagger}$) and the new error $UEU^{\dagger}$ are correlated. Nevertheless, their joint distribution can be approximated by a product distribution, as we now show. 

\textbf{Notation:} In this subsection, we'll use the symbol $\mathcal{P}_{w}$ to denote the set of Paulis with Hamming weight $w$, and the symbol $\algo U_w$ to denote the uniform distribution over $\mathcal{P}_{w}$. 

\begin{lemma}[Re-randomization of the code and error\label{lem:rerandomize_codeerror}]
Fix the $\lsn$ instance
$(S, E\ket{\overline{\psi_x}}^S)$, for some
stabilizer subgroup $S\in \Stab(n,k)$, error $E\in \bar{\algo P}_n$ and secret $x\in \{0,1\}^k$. Acting on the state with an $n$-qubit Clifford unitary $U \sim \plc_n$ produces a new state which can be interpreted as an encoding of $x$ under the stabilizer code $USU^{\dagger}$, acted upon by error $UEU^{\dagger}$ (\Cref{eq:correlatedcodeerror}).

Suppose that $E$ has bounded weight $w= O(\log^c n)$ for some integer $c > 0$. Then, the joint distribution of the new error and new code, $(UEU^{\dag},USU^{\dag})$ is within total variation distance 
$1- 1/O(n^{\log^c n})$ of the product distribution $\algo U_w \times \mathrm{Unif}(\{USU^\dag\}_{U \in \plc_n})$. 
\end{lemma}
\begin{proof}
We analyze the joint distribution of $(UEU^{\dag},USU^{\dag})$, for $U \sim \plc_n$.
    We may regard $U:\calP_n\times \Stab(n,k) \rightarrow \calP_n\times \Stab(n,k)$ as a map that acts on the pair of initial error and stabilizer subgroup $(E,S)\in  \calP_n\times \Stab(n,k)$ via the group action
    \begin{equation}
        U *(E,S) := (UEU^{\dag}, USU^{\dag}).
    \end{equation}
    We claim that for $U\sim \plc_n$, the joint distribution of $(UEU^{\dag},USU^{\dag})$ is uniform over the $\plc_n$ orbit of $(E,S)$, namely the set
    \begin{equation}
        \plc_n *(E,S):= \{(P,T)\in  \calP_n\times \Stab(n,k)\, :\, \exists \,U\in \plc_n \,\, \mathrm{ s.t. }\,\, P=UEU^{\dag}, T=USU^{\dag}\}.
    \end{equation}
    To see this, note that for any $(P_1,T_1), (P_2,T_2)\in  \plc_n *(E,S)$, there must exist $V\in \plc_n$ such that $VP_1 V^{\dag} = P_2$ and $VT_1V^{\dag} = T_2.$ This is because $\plc_n$ is a group (Lemma \ref{lem:plcgroup}), and by the definition of $\plc_n *(E,S)$, there exist $V_1, V_2 \in \plc_n$ such that $$(P_1,T_1) = (V_1EV_1^{\dag},V_1SV_1^{\dag}), (P_2,T_2) = (V_2EV_2^{\dag},V_2SV_2^{\dag}).$$ Thus, $V$ is exactly $V_2V_1^{\dag}$. As a result, we conclude that 
    \begin{align}
        &\Pr_{U\sim \plc_n}[UEU^{\dag}=P_1 \, \cap USU^{\dag} = T_1]\\
        &= \Pr_{U\sim \plc_n}[VUEU^{\dag}V^{\dag}=VP_1V^{\dag} \, \cap VUSU^{\dag}V^{\dag} = VT_1V^{\dag}]\\
        &= \Pr_{U\sim \plc_n}[VUEU^{\dag}V^{\dag}=P_2 \, \cap VUSU^{\dag}V^{\dag} = T_2]\\
        &= \Pr_{U\sim \plc_n}[UEU^{\dag}=P_2 \, \cap USU^{\dag} = T_2].
    \end{align}
    We will now argue that the uniform distribution $\mathsf{Unif}(\plc_n *(E,S))$ is a good approximation to the distribution $\mathcal{U}_w\times \mathrm{Unif}(\{USU^\dag\}_{U \in \plc_n})$ by bounding the total variation distance between the two. First, note that the total variation distance between the uniform distribution on finite sets $\algo X$ and $\algo Y\subseteq \mathcal{X}$ takes the following simple form:
\begin{equation}
    \delta_{\mathsf{TV}} \left(\text{Unif}(\algo X), \text{Unif}(\algo Y)\right) = \frac{1}{|\algo X|}(|\algo X|-|\algo Y|) = 1-\frac{|\algo Y|}{|\algo X|}.
\end{equation}
Letting set $\algo X$ be the set of all tuples $\{(P, USU^{\dagger})\}_{P\in \mathcal{P}_w, U\in \plc_n}$ and set $\algo Y = \plc_n *(E,S)$, we have that  
\begin{align}
    &\delta_{\mathsf{TV}}\left( \mathcal{U}_w \times \mathrm{Unif}(\{USU^\dag\}_{U \in \plc_n}),\mathsf{Unif}(\plc_n *(E,S))\right)\\
    &= 1- \frac{|\plc_n *(E,S)|}{|\mathcal{P}_{w}| \cdot|\{USU^{\dag}\}_{U\in \plc_n}|}\\
    &\leq 1-\frac{1}{|\mathcal{P}_{w}|} \,= \, 1-\frac{1}{2\binom{n}{w}3^w}.
    \end{align}
where we have used the fact that $|\plc_n *(E,S)| \geq |\{USU^{\dag}\}_{U\in \plc_n}|$, because for every distinct $T\in \{USU^{\dag}\}_{U\in \plc_n}$, $\exists V\in \plc_n$ such that $T = VSV^{\dagger}$, and so there is at least one element with $(VEV^{\dagger},T) \in \plc_n *(E,S)$. Plugging in the assumption that $w = O(\log^c n)$, we conclude that $2\binom{n}{w}3^w \leq 2\left(\frac{3ne}{w}\right)^w = O(n^{\log^c n})$, which proves the claim.
\end{proof}

\subsection{Worst-Case to Average-Case Reduction}

In this section, we formally state our worst-case to average-case reduction for a variant of the LSN problem.
Specifically, we show how an appropriate average-case solver allows us to solve worst-case LSN instances $(S, E\ket{\overline{\psi_x}}^S)$. 
To this end, we assume that $S$ is a fixed (and worst-case choice of stabilizer) and we consider the average-case problem $\lsn_{n,k,\mathcal{N},\mathcal{S},\mathcal{I}}$ with respect to the following set of distributions:
\begin{itemize}
    \item $\mathcal{N}$ is the uniform distribution $\algo U_w$ over $n$-qubit Pauli errors of weight precisely $w$.

    \item $\mathcal{S}$ is the uniform distribution $\mathrm{Unif}(\{USU^\dag\}_{U \in \plc_n}$ over stabilizers $\{USU^\dag\}_{U \in \plc_n}$.

    \item $\mathcal{I}$ is the uniform distribution over bit strings $x \in \bit^k$.
\end{itemize}
While our worst-case to average-case reduction only applies to a highly specialized variant (in particular, not to the standard variant) of LSN in the quasi-polynomial hardness regime\footnote{See \Cref{remark:quasi} for a definition of quasi-polynomial hardness for the LSN problem.}, it nevertheless results in a meaningful reduction. Concretely, it allows us to solve a worst-case problem with inverse-quasi-polynomial success probability whenever we have a sufficiently good solver for an average-case version of the problem. 

\begin{theorem}[Worst-case to average-case reduction]\label{thm:Reduction_LSN}
Let $n,k \in \mathbb{N}$ with $n = \poly(k)$ and let 
$(S, E\ket{\overline{\psi_x}}^S)$ be any worst-case instance, for some stabilizer subgroup $S\in \Stab(n,k)$, error $E\in \bar{\algo P}_n$ of weight $w=O(\log^c n)$ for $c > 0$, and secret $x\in \{0,1\}^k$.
Suppose there exists an algorithm $\mathcal{A}$ that runs in time $T$ and solves the average-case problem $\lsn_{n,k,\mathcal{N},\mathcal{S},\mathcal{I}}$ (implicitly depending on $S$) with probability $1-\epsilon$.
Then, there exists an algorithm $\mathcal{B}$ which runs in time $\poly(k,T)$ and solves the worst-case instance $(S, E\ket{\overline{\psi_x}}^S)$ with probability at least $1/O(n^{\log^c n})-\epsilon$.
\end{theorem}

\begin{proof}
By assumption, the instance to the worst-case problem is of the form 
$(S, E\ket{\overline{\psi_x}}^S)$,
for some stabilizer subgroup $S\in \Stab(n,k)$, error $E\in \bar{\algo P}_n$ of weight $w=O(\log^c n)$, $c>0$, and secret $x\in \{0,1\}^k$.
We will now give a reduction $\algo B$ which transforms the given sample $(S, E\ket{\overline{\psi_x}}^S)$ into a new sample which approximates the average-case instance for the $\lsn_{n,k,\mathcal{N},\mathcal{S},\mathcal{I}}$ problem.  Our reduction $\algo B$ uses the solver $\algo A$ and proceeds as follows:
\begin{enumerate}
    \item $\algo B$ samples a random logical operator $\overline{X^{u}}\sim \algo L_X(S)$, for $u \in \bit^k$.

    \item $\algo B$ samples a random unitary $U \sim \plc_n$ from the $\plc_n$ ensemble.

    \item $\algo B$ runs the solver $\algo A$ for the $\lsn_{n,k,\mathcal{N},\mathcal{S},\mathcal{I}}$ problem on input
$$
(USU^\dag, U \overline{X^{u}} E\ket{\overline{\psi_x}}^S)
$$
to obtain a string $x' \in \bit^k$. Then, the reduction $\algo B$ outputs $x' \oplus u$.
\end{enumerate}
In other words, the reduction $\algo B$ applies the $n$-qubit re-randomizing unitary consisting of $R = U \overline{X^{u}}$ to the initial noisy codeword $E\ket{\overline{\psi_x}}^S$. 

We now analyze the probability that $\algo B$ succeeds at recovering the secret $x \in \bit^k$. First, we use our insights from \Cref{lem:rerandomize_secret} and \Cref{lem:rerandomize_codeerror} to argue that
$$
(USU^\dag, U \overline{X^{u}} E\ket{\overline{\psi_x}}^S)= (USU^\dag, (UEU^\dag)\ket{\overline{\psi_{x \oplus u}}}^{USU^\dag}).
$$
In addition, we know from \Cref{lem:rerandomize_secret} that the distribution of the secret is precisely $\algo I$, and we know from \Cref{lem:rerandomize_codeerror} that the distribution of the stabilizer subgroup and error of the resulting state is within total variation distance at most $1-1/O(n^{\log^c n})$ of the product distribution $\algo S \times \algo N$. Therefore, by the strong convexity of the trace distance (\Cref{lem:strong-convexity}), we know that $\algo B$ succeeds with probability at least $1/O(n^{\log^c n})-\epsilon$.

We now argue that it takes $\poly(k,T)$ time to perform the reduction. First, we can invoke \Cref{lem:sim} to argue that a random logical Pauli can be computed in time $O(n^2)$.
Then, from \Cref{lem:sim}, it follows that the reduction $\algo B$ can compute a classical description of the stabilizer subgroup $U S U^\dag$ in time polynomial in $k$ and $n$, since $U \sim \plc_n$ is an efficient Clifford operator. 
Because $\algo A$ runs in time $T$, this completes the proof.
\end{proof}

While we do not explicitly carry out the proof, we remark that a similar worst-case to average-case reduction also applies to the multi-shot variant of LSN: the reduction can simply apply a fresh re-randomizing unitary for every block of $n$ qubits.

\section{Complexity of Learning Stabilizers with Noise}
\label{sec:complexity}

In previous sections, we have postulated a {\em lower bound} for the time-complexity of $\lsn$ (\Cref{def:LSN}); that is, there do not exist efficient quantum algorithms that solve the problem. In this section, we upper-bound the complexity of $\lsn$, placing it within the constellation of unitary synthesis problems proposed in \cite{rosenthal2021interactiveproofssynthesizingquantum, bostanci2023unitarycomplexityuhlmanntransformation}. Specifically, we show that a variant of the LSN problem, for a worst-case choice of non-degenerate code, is contained in $\mathsf{avgUnitaryBQP}^{\mathsf{NP}}$, a (distributional and oracle) unitary complexity class which was recently defined by Bostanci et al.~\cite{bostanci2023unitarycomplexityuhlmanntransformation}.

\subsection{A Review of Unitary Complexity}
This subsection reprises some problems and complexity classes introduced in \cite{rosenthal2021interactiveproofssynthesizingquantum,bostanci2023unitarycomplexityuhlmanntransformation} as a way of giving background for our complexity upper bound.
\paragraph{Unitary synthesis problems.} Many quantum problems whose output is a quantum state or unitary fall outside the purview of traditional complexity theory. Some examples include implementing Hamiltonian time evolution and state preparation tasks. All of these tasks have the flavor of {\em preparing a target unitary} upon input of some classical description of the target. This led to the formalization of \emph{unitary synthesis problems}:

\begin{definition}[Unitary synthesis problems\label{def:unitarysyn}]
    A unitary synthesis problem is given by a sequence $\mathscr{U}=$ $\left(U_x\right)_{x \in\{0,1\}^*}$ of partial isometries.\footnote{See \Cref{sec:prelims} for a formal definition of partial isometries.}
\end{definition}
We may understand $x \in\{0,1\}^{\ast}$ as the way that the particular target partial isometry is specified to the algorithm that solves the problem. 

In the definition above, we call $x$ the instance of the problem and $U_x$ the transformation of $\mathscr{U}$ corresponding to $x$. The goal of an algorithm handed an instance $x$ of a unitary synthesis problem is then to implement a quantum channel $C_x$ which approximates a channel completion of the target unitary $U_x$ in diamond norm. In fact the algorithm must accomplish this for all problem instances, $x\in \{0,1\}^{\ast}$. One could consider various metrics for how well the algorithm's output $C_x$ approximates the target; a ``worst-case" measure is to require the existence of a channel completion $\Phi_x$ of $U_x$ such that
\begin{equation}
\left\|C_x-\Phi_x\right\|_{\diamond} \leq \delta(|x|)\quad \forall x\in \{0,1\}^{\ast}.
\end{equation}
The strict requirement of diamond-norm approximation makes this a ``worst-case" measure of closeness: it says there must exist a channel completion of the target unitary such that for any choice of registers to trace out, tracing out those registers of the channel completion still gives a channel that is well-approximated by the channel $C_x$ after tracing out the same registers. 
\subparagraph{Average-case/distributional unitary synthesis problems:}
It is not necessary for us to use this strict notion of approximation; we will instead be using an {\em average-case} notion of approximation captured by a {\em distributional (or average-case) unitary synthesis problem}. Here, in addition to a specifying a target partial isometry, we also specify an input state $\ket{\psi_x}$ and a register of this state on which the partial isometry is going to act, and we only care about closeness with respect to this register. We call the register that $U_x$ (or its channel completion) acts on the \emph{quantum input} to the unitary synthesis problem.

\begin{definition}[Distributional unitary synthesis problem]
\label{def:dist-usynth}
A \emph{unitary synthesis problem} is given by a sequence $\mathscr{U} = (U_{x})_{x \in \{0,1\}^*}$ of partial isometries. 
We say that a pair $(\mathscr{U}, \Psi)$ is a \emph{distributional unitary synthesis problem} if $\mathscr{U} = (U_x)_{x \in \{0,1\}^*}$ is a unitary synthesis problem with $U_x \in \linear(\mathcal{H}_{\mathsf{A}_x}, \mathcal{H}_{\mathsf{B}_x})$ for some registers $\mathsf{A}_x\mathsf{B}_x$, and $\Psi = (\ket{\psi_x})_{x \in \{0,1\}^*}$ is a family of bipartite pure states on registers $\mathsf{A}_x \mathsf{R}_x$. We call $\ket{\psi_x}$ the \emph{distribution state} with \emph{target register} $\mathsf{A}_x$ and \emph{ancilla register} $\mathsf{R}_x$. 
\end{definition}

\begin{definition}[Average-case implementation of distributional unitary synthesis] \label{def:avg_case_error}
Let $(\mathscr{U},\Psi)$ denote a distributional unitary synthesis problem, where $\mathscr{U} = (U_x)_{x \in \{0,1\}^*}$ and $\Psi = (\ket{\psi_x})_{x \in \{0,1\}^*}$, and let $\epsilon: \mathbb{N} \to \mathbb{R}$ be a function. Let $C = (C_x)_{x \in \{0,1\}^*}$ denote a family of quantum circuits, where $C_x$ implements a channel whose input and output registers are the same as those of $U_x$. We say that \emph{$C$ implements $(\mathscr{U},\Psi)$ with \textbf{average-case error} $\epsilon$} if, for all sufficiently long $x \in \{0,1\}^*$, there exists a channel completion $\Phi_x$ of $U_x$ such that
\[
\delta_{\mathsf{TD}}\Big ( (C_x \ot \id)(\psi_x), \, (\Phi_x \ot \id)(\psi_x) \Big ) \leq \epsilon(|x|) \,,
\]
where the identity channel acts on the ancilla register of $\ket{\psi_x}$.
\end{definition}

\paragraph{Unitary/state complexity classes.}
Unitary complexity problems can be grouped into complexity classes. These complexity classes are organized based on the amount of computational resources needed to perform state transformations. As important background, let us informally introduce the unitary complexity class $\mathsf{unitaryBQP}$. Analogous to $\mathsf{BQP}$ which is the set of all decision problems that can be solved by a polynomial-time quantum computer with at most $1/3$ probability of error, $\mathsf{unitaryBQP}$ is the set of all partial isometries that can be approximately applied in polynomial time in their description length. That is, it is the set of all sequences of unitary operators $\left(U_x\right)_{x \in\{0,1\}^*}$ where there is a polynomial-time quantum algorithm $A$ that, given an instance $x \in\{0,1\}^*$ and a quantum system B as input, (approximately) applies $U_x$ to system B. Here the input system B could contain any state, even part of a state on a larger system. 

\subparagraph{Average-case/distributional unitary complexity classes:}
We will eventually be interested in the average-case version of $\mathsf{UnitaryBQP}$, which is $\mathsf{avgUnitaryBQP}$. Recall that to go from worst-case to distributional, or average-case unitary synthesis problems, we introduce {\em distribution states} -- input states whose specific registers we will be implementing the desired unitary transformation on. In order to properly define $\mathsf{avgUnitaryBQP}$, therefore, we must introduce the state complexity class $\mathsf{stateBQP}$ which was introduced in~\cite{rosenthal2021interactiveproofssynthesizingquantum}. Intuitively, this class contains sequences of quantum states that require polynomial time to be synthesized. 

\begin{definition}[$\mathsf{stateBQP}$] \label{def:stateclasses}
	Let $\delta: \mathbb{N} \to [0,1]$ be a function. Then, $\mathsf{stateBQP}_{\delta}$ is the class of all sequences of density matrices $(\rho_x)_{x \in \bit^*}$ such that each $\rho_x$ is a state on $\poly(|x|)$ qubits, and there exists a time-uniform family of general quantum circuits $(C_x)_{x \in \bit^*}$ such that, for all sufficiently long $x \in \bit^*$, the circuit $C_x$ takes no inputs and $C_x$ outputs a density matrix $\sigma_x$ such that 
	\[
	\delta_{\mathsf{TD}}(\sigma_x, \rho_x) \leq \delta(|x|)\,.
	\]
	We define
	\[
	    \mathsf{stateBQP} = \bigcap_{c \in \mathbb{N}} \mathsf{stateBQP}_{n^{-c}}.
	\]
\end{definition}

With this definition in hand, we can define $\mathsf{avgUnitaryBQP}$ as the set of polynomial-time solvable distributional unitary synthesis problems, with the restriction that their input state is in $\mathsf{stateBQP}$ (i.e. is polynomial-time preparable).
\begin{definition}[$\mathsf{avgUnitaryBQP}$]\label{def:avgunitaryBQP_avgunitaryPSPACE}
Let $\epsilon: \mathbb{N} \to \mathbb{R}$ be a function. Define the unitary complexity class $\mathsf{avgUnitaryBQP}_\epsilon$ to be the set of distributional unitary synthesis problems $\Big( \mathscr{U} = (U_x)_{x \in \{0,1\}^*}, \Psi= (\ket{\psi}_x)_{x \in \{0,1\}^*} \Big)$ where $\Psi \in \mathsf{stateBQP}$ and there exists a uniform polynomial-time quantum algorithm $C$ that implements $(\mathscr{U},\Psi)$ with \textbf{average-case error} $\epsilon$.

We define
\[
    \mathsf{avgUnitaryBQP} = \bigcap_{c \in \mathbb{N}} \mathsf{avgUnitaryBQP}_{n^{-c}}.
\]
\end{definition}

\subparagraph{Average-case oracle unitary complexity classes.}

We also consider reductions between (average-case) variants of unitary complexity classes. Hence, we use oracular variants of (average-case) unitary complexity classes, which behave analogously as their classical counterparts. We refer to Section 3.3 of~\cite{bostanci2023unitarycomplexityuhlmanntransformation} for more details.

\subsection{Complexity Upper Bound}

We conclude this section by showing that (a variant of) the average-case  $\mathsf{LSN}_{n,k,\mathcal{D}_p^{\otimes n}}$ problem from \Cref{def:LSN}---for a worst-case choice of non-degenerate stabilizer code---is contained in the complexity class $\mathsf{avgUnitaryBQP}^{\mathsf{NP}}$. Note that this requirement that the stabilizer code is non-degenerate is essentially without loss of generality: the quantum Gilbert-Varshamov bound guarantees that a random stabilizer code is non-degenerate and has large distance with overwhelming probability.

In the following proof, the stabilizer code associated with Clifford $C \in \Cliff_n$ is assumed to be fixed and non-degenerate. When we refer to the ``average-case" problem, we mean average-case over choice of error and secret. This is sound in the sense that the stabilizer code (and its accompanying encoding circuit $C$) is known to the $\lsn$ algorithm, whereas the error and secret are not.

\begin{theorem}[Complexity of LSN]\label{thm:complexity_LSN}
Let $k \in \mathbb{N}$ and $n \geq 8k$.
Let
$\mathcal{D}_p^{\otimes n}$ be the $n$-qubit depolarizing channel, for some constant $p\in(0,0.05)$. Then, the
$\mathsf{LSN}_{n,k,\mathcal{D}_p^{\otimes n}}$ problem with a worst-case choice of non-degenerate stabilizer code with distance $d>3np$ is contained in the distributional (and oracle) 
unitary synthesis class $\mathsf{avgUnitaryBQP}^{\mathsf{NP}}$.    
\end{theorem}

\begin{proof}
Let us fix $C \in \Cliff_n$ as the encoding Clifford for a worst-case non-degenerate stabilizer code with distance $d>3np$. 
As we have previously observed in \Cref{sec:unique-sol}, we can interpret the $\mathsf{LSN}_{n,k,\mathcal{D}_p^{\otimes n}}$ problem
as a distributional unitary synthesis problem. 
Specifically, the instance consists of a classical description $(1^k, Q^0, Q^1)$ together with an input $\ket{\Psi}_{\mathsf{AB}} = \ket{Q^0}_{\mathsf{AB}}$. Here, $(Q^0, Q^1)$ are a pair of quantum circuits that, upon input $\ket{0}$ prepare the states
\begin{align*}
\ket{Q^0}_{\mathsf{AB}} &= 
\sqrt{2^{-k}}\sum_{x} \sum_{E_a} \sqrt{\Pr_{E_a \sim \mathcal{D}_{p}^{\otimes n} }[E_a]} \,\Big(\ket{x} \otimes \ket{a}\Big)_{\mathsf{A}} \otimes \left(E_{a} \, C (\ket{0^{n-k}} \otimes \ket{x}) \otimes \ket{0}\right)_{\mathsf{B}}\\
\ket{Q^1}_{\mathsf{AB}} &= \sqrt{2^{-k}}\sum_{x} \sum_{E_a} \sqrt{\Pr_{E_a \sim \mathcal{D}_{p}^{\otimes n} }[E_a]} \,\Big(\ket{x} \otimes \ket{a}\Big)_{\mathsf{A}} \otimes \left(\ket{0^{n-k}} \otimes \ket{x} \otimes \ket{a}\right)_{\mathsf{B}}.
\end{align*}
while $\ket{\Psi}_{\mathsf{AB}} = \ket{Q^0}_{\mathsf{AB}}$. We note that $\ket{Q^0}_{\mathsf{AB}}$ is simply a purification of the register containing the quantum part of the input to the $\mathsf{LSN}_{n,k,\mathcal{D}_p^{\otimes n}}$ problem.

We now claim that the following $\mathsf{BQP}$ machine with access to an $\mathsf{NP}$ oracle can synthesize the appropriate unitary transformation in order to solve the $\mathsf{LSN}_{n,k,\mathcal{D}_p^{\otimes n}}$ problem:
\begin{enumerate}
    \item Use the ancillary register in the state $\ket{0}$ to coherently compute a stabilizer syndrome with respect to the first half of the $\mathsf{B}$ register.

    \item Call the $\mathsf{NP}$ oracle in superposition to recover a classical description of the corresponding Pauli error (up to sign) that corresponds to the stabilizer syndrome.\footnote{Note that the $\mathsf{NP}$ oracle solves decision problems, whereas in our case we need to solve a search problem in order to extract the correct error from a given syndrome. Fortunately, a standard search-to-decision reduction applied to syndrome decoding problems suffices to get around this issue.}

    \item Controlled on the ancilla register, coherently apply the corresponding Pauli correction on the first half of $\mathsf{B}$ followed by the inverse encoding circuit $C^\dag$.
\end{enumerate}
To complete the proof, we can now essentially follow the same analysis as in \Cref{sec:unique-sol}. First,
with overwhelming probability over the choice of error, it follows from the quantum Gilbert-Varshamov bound (and our choice of parameters) that the $\mathsf{B}$ register of $\ket{Q^0}_{\mathsf{AB}}$ carries a unique error syndrome associated with an error Pauli. Hence, such an error can always be uniquely recovered from a given syndrome via the $\mathsf{NP}$ oracle.

\end{proof}

\section{Applications}

We now discuss two applications which rest on the hardness of the LSN problem.

\subsection{Learning From Quantum Data}\label{sec:learning}
The hardness of $\lpn$ implies conditional computational lower bounds on many classical learning tasks. Here we show that the hardness of $\lsn$ would also imply (quantum) computational lower bounds on the task of {\em learning from quantum data} \cite{Caro24}, by studying a special case known as {\em learning state preparation processes}. As a bonus, we also provide an {\em upper} bound on the complexity of this task; specifically, we analyze it using the newly developed framework of unitary synthesis problems. We hope that this contribution paves the way for future work on the complexity of quantum learning tasks.

What is the quantum generalization of the hugely fruitful classical framework of Probably Approximately Correct (PAC) learning \cite{Valiant84}? This question has received much attention recently \cite{Caro_2021, FQR, chung2021sampleefficientalgorithmslearning}, culminating in the formulation of a powerful general framework known as `learning from quantum data' by Ref. \cite{Caro24}, which encompasses the settings of PAC learning quantum states \cite{AaronsonPAC}, PAC learning from quantum examples \cite{BJ, AdW17}, variational quantum machine learning \cite{BPP21} and so on. 
\paragraph{Lower bounding the complexity of learning from quantum data.} We lower bound the complexity of PAC-learning state-preparation processes, a special case of learning from quantum data. This implies a bound on the more general task. This setting is intended to model the scenario in which an experimentalist lacks significant control over the inputs to a process occurring in nature that she nevertheless wishes to understand. Here the process takes classical inputs (e.g. time, temperature, magnetic fields) to quantum states (e.g. electromagnetic radiation collected from astronomical events). The learner observes random input, output pairs $(x,\rho(x))$, and in particular cannot query the process at identical input points. 

The learner's task is to output an estimate of $\rho$. To measure how far the learner's output is from the true $\rho$, we introduce the notion of {\em risk} (relative to $\rho$) for any process $h: \mathcal{X} \rightarrow L(\mathcal{H}_d)$:
\begin{equation}\label{eq:risk}
    R_\rho(h):=\mathbb{E}_{x \sim \mathcal{D}}[L(\rho(x), h(x))],
\end{equation}
where $L: L(\mathcal{H}_d)\times L(\mathcal{H}_d) \rightarrow \mathbb{R}$ is the trace distance. 

\begin{definition}[Learning State Preparation processes\label{def:learnprocess}]
    Let unknown process $\rho: \mathcal{X} \rightarrow L(\mathcal{H}_d)$ be a map that assigns to points in a classical input space $\mathcal{X}$ a corresponding quantum state, $\mathcal{C} = \{h : \mathcal{X} \rightarrow L(H)\}$ be a set of hypotheses for what $\rho$ could be, and $\mathcal{D}: \mathcal{X} \rightarrow [0,1]$ be a known distribution over the inputs. 
    
    The Learning State Preparation processes problem is to output some hypothesis $h\in \mathcal{C}$ satisfying $R_\rho(h)\leq \epsilon,$ given as input copies of a classical-quantum state $\sigma$, 
\begin{equation}\label{eq:learningquantumprocesses_input}
    \sigma =\underset{x \sim \mathcal{D}}{\mathbb{E}}[|x\rangle\langle x| \otimes \rho(x)].
    \end{equation}

    We say that the learner solves the problem of Learning State Preparation processes with sample complexity $m$ if, given input $\sigma^{\otimes m}$, it succeeds with probability at least $1-\delta$ to output a $h$ satisfying \Cref{eq:risk} where $\delta$ is constant.
\end{definition}
Due to concentration, to output $h$ minimizing the risk it suffices to minimize the empirical risk i.e., the average loss computed on the examples $\left(x_i, \rho\left(x_i\right)\right)_{i=1}^n$ :
$$
\hat{R}_\rho(h):=\frac{1}{m} \sum_{i=1}^m L(\rho(x_i), h(x_i)).
$$
The sample complexity of empirical risk minimization for Learning State Preparation processes was resolved by Ref. \cite{FQR}. They defined a quantum version of empirical risk minimization and showed that it constitutes a learning algorithm:
\begin{theorem}[Quantum Empirical Risk Minimization (Theorem 3 of \cite{FQR})\label{thm:QERM}]
There exists a learner for state preparation processes, which, for any $\rho$ not necessarily within $\mathcal{C}$, with probability $1-\delta$ outputs 
$\sigma^{\ast}\in \mathcal{C}$ that approximately minimizes the empirical risk:
\begin{equation}
    \hat{R}_{\rho}(\sigma^{\ast}) \leq 3 \min \hat{R}_{\rho}(\sigma_i) + 4\epsilon
\end{equation}
with sample complexity \begin{equation}\label{eq:scomp}
m = \tilde{O}\left(\frac{\log d \log \frac{1}{\delta} \max \left(\log \frac{|\mathcal{C}|}{\delta}, \log^2(e |\mathcal{C}|)\right)}{\epsilon^5}\right).
\end{equation}
\end{theorem}

Quantum empirical risk minimization \Cref{thm:QERM} is sample-efficient whenever $d = O(\exp(n))$ and $|\mathcal{C}| = O(\exp(n))$. However, the time complexity of learning was not addressed in Ref. \cite{FQR}, and thus we attempt to do it in this work. In a nutshell, the following theorem says that an algorithm that could learn state-preparation processes can also decode exactly in the presence of noise.

\begin{theorem}[Learning state preparation processes can be sample- but not time-efficient\label{thm:notefficient}]
Let $p\in (0,1/2)$ such that
\begin{equation}\label{eq:bounddelta}
    H(3p)+3\log_2(3) p + k/n< .99 -\frac{\log(3/2)}{n}.
\end{equation} 
Conditioned on the hardness of $\mathsf{MSLSN}_{m,n,k,\mathcal{D}_p^{\otimes n}}$, there is no time-efficient algorithm for learning state preparation processes, even when only $1/poly(n)$ success probability is required and the unknown process $\rho$ is guaranteed to be in the concept class $\mathcal{C}$ (the ``proper" learning setting). 
\end{theorem}
\begin{proof}
We observe that any algorithm for proper learning state preparation processes could also solve $\mathsf{MSLSN}_{m,n,k,\mathcal{D}_p^{\otimes n}}$, by choosing
\begin{align}\label{eq:translationtomslsn}
    \mathcal{X} &:= \{S\}_{S\in \text{Stab}(n,k)} \qquad \text{(Classical input domain)}\\
    \mathcal{D} &:= \mathsf{Unif}(\{S\}_{S\in \text{Stab}(n,k)}) \qquad \text{(Distribution over inputs)}\\
    \rho_z(S) &:= \mathcal{D}_p^{\otimes n}(U_{\text{Enc}}^S(\ket{z}\bra{z})) \qquad \text{(Unknown map to be learned)}\\
    \mathcal{C} &:= \{\rho_z\}_{z\in \{0,1\}^k} \qquad \text{(Concept class)};
\end{align}
Moreover $\epsilon$ in \Cref{def:learnprocess} must be chosen so that the algorithm outputs the exact concept instead of a mere approximation; this is the meaning of decoding (solving $\mslsn$). $\epsilon$ must be such that no other concept lies in the $\epsilon$-ball of the correct solution:
\begin{align}
    \epsilon &\leq \min_{z_1,z_2\in \{0,1\}^k} \frac{1}{m} \sum_{i=1}^m \delta_{\mathrm{tr}}(\rho_{z_1}(S_i), \rho_{z_2}(S_i)),
\end{align}
and we proceed to show that a constant-$\epsilon$ learner for state preparation processes suffices. Using the convexity of trace distance (with an argument similar to that in the proof of \Cref{thm:multishotdecoding}), it suffices to tackle the case when the noise on each copy is given by the truncated depolarizing channel $\tilde{\mathcal{D}}^{(3/2np)}$ at the cost of an exponentially small increase in failure probability $me^{-\frac{np}{12}}$. 
Therefore it suffices to take
\begin{align}
    \epsilon &\leq \min_{x,y\in \{0,1\}^k} \frac{1}{m} \sum_{i=1}^m \delta_{\mathrm{tr}}(\tilde{\mathcal{D}}^{(\frac{3np}{2})}(\proj{\overline{\psi_x}}^{S_i}), \tilde{\mathcal{D}}^{(\frac{3np}{2})}(\proj{\overline{\psi_y}}^{S_i}))\\
    &\leq \min_{x,y\in \{0,1\}^k} \frac{1}{m} \sum_{i=1}^m \sqrt{1-\mathsf{F}(\tilde{\mathcal{D}}^{(\frac{3np}{2})}(\proj{\overline{\psi_x}}^{S_i}),\tilde{\mathcal{D}}^{(\frac{3np}{2})}(\proj{\overline{\psi_y}}^{S_i}))^2}
\end{align}
where the second inequality follows from Fuchs-van de Graf. We further recall from \Cref{thm:GVbound} that with probability at least $1-\delta_1, \delta_1 = 3 np \cdot 2^{n H(3p) } \cdot 3^{3np} \cdot 2^{-n+k}$ over choice of $S_i\in \Stab(n,k),$ 
\begin{equation}
    1-\mathsf{F}(\tilde{\mathcal{D}}^{(\frac{3np}{2})}(\proj{\overline{\psi_x}}^{S_i}),\tilde{\mathcal{D}}^{(\frac{3np}{2})}(\proj{\overline{\psi_y}}^{S_i}))^2 = 1;
\end{equation}
call $S_i$ ``good" if this is true. (The condition \Cref{eq:bounddelta} in the Theorem statement is necessary to make $\delta_1\leq1$.) 

By H\"{o}ffding's inequality, with probability $1-\exp(-\delta_1^2n/2),$ the number of samples featuring $S_i$ that are ``good" is at least $(1-3\delta_1/2)m.$ Therefore any learner of state preparation processes up to any constant $\epsilon$,
\begin{equation}\label{eq:const}
    \epsilon \leq \frac{1}{m}((1-3\delta_1/2)m) = 1-3\delta_1/2,
\end{equation}
will solve the $\mslsn$ problem with a total failure probability of at most
\begin{equation}
\delta + \exp(-\delta_1^2n/2) + \exp(-\frac{np}{12}) = \delta + O(\exp(-n)).
\end{equation}
\end{proof}

\begin{remark}\Cref{thm:notefficient} should be compared to the hardness result of Ref \cite{AGS21}. There, the concept class to be learned is the set of quantum circuits that compute classical Boolean functions $c$ output by $\mathsf{AC}^0$ and $\mathsf{TC}^0$ circuits, and in one of their learning models, the learner is given access to quantum examples of the form
\begin{equation}\label{eq:examp}
    \sum_x \sqrt{\mathcal{D}(x)} \ket{x,c(x)}.
\end{equation}
The goal is to output a hypothesis $h$ such that $\operatorname{Pr}_{x \sim \mathcal{D}}[h(x) \neq c(x)] \leq \varepsilon$. Ref \cite{AGS21} showed that conditioned on the quantum hardness of $\mathsf{RingLWE}$ and $\mathsf{LWE}$, no such polynomial-time learner can exist. This also implies the non-existence of polynomial-time quantum learners for general state-preparation processes, corresponding to the case when the unknown process is classical -- maps to binary labels $\kb{0},\kb{1}$ and is computable by a classical $\mathsf{TC}^0$ or $\mathsf{AC}^0$ circuit. To complete the reduction, note that measuring the first $n-1$ qubits of a quantum example (\Cref{eq:examp}) yields the learner's input (\Cref{eq:learningquantumprocesses_input}). 

Given that the setting of \cite{AGS21} constrains quantum learners of classical functions, it does not say much about learning from `natively quantum' data. For example, when the concept class consists of noise channels, it is unclear how to write their purifications as Boolean functions output by $\mathsf{TC}^0$ or $\mathsf{AC}^0$ circuits. We expect $\lsn$ to be more useful to constrain learning in such situations.
\end{remark}

\subsection{Quantum Bit Commitments}
\label{sec:commitments}

$\lpn$ is fundamental to classical cryptography. In this section, we show that our LSN assumption also has applications in quantum cryptography; we use it to construct a (statistically hiding and computationally binding) quantum bit commitment scheme.

Bit commitment is a fundamental primitive in cryptography with a multitude of applications, ranging from secure coin flipping, to zero-knowledge proofs, and secure computation.
In classical cryptography, commitment schemes can be constructed from any one-way function~\cite{Naor2003}. In quantum cryptography, potentially even weaker and inherently quantum assumptions suffice; for example, the existence of so-called pseudorandom states~\cite{Kretschmer21,ananth2022cryptography,morimae2022quantum}. 

A (canonical) quantum bit commitment scheme~\cite{yan2023general}
is a pair of efficient quantum circuits $(Q^0,Q^1)$ which output two registers: a ``commitment register'' $\mathsf C$  and a ``reveal register'' $\mathsf R$. To commit to a bit $b \in \{0,1\}$, the sender prepares the state $\ket{Q^b}_{\mathsf{CR}} = Q^b \ket{0}_{\mathsf{CR}}$ and sends the register $\reg C$ to the receiver. In the ``reveal phase'', the sender simply reveals the bit $b$ together with register $\mathsf R$. The receiver accepts if the inverse unitary $Q^{b,\dag}$ applied to registers $\mathsf{CR}$ yields $\ket{0}$ when measured in the computational basis. 
    In terms of security, there are to important properties that we associate with a commitment scheme. 
First, the \emph{statistical hiding} property ensures that the commitment register (information-theoretically) \emph{hides} the committed bit $b$, i.e.~after the commit phase, the receiver cannot guess what bit the sender committed to.
Second, the \emph{computational binding} property ensures that, after the commitment phase, it is computationally intractable for the sender to change the bit $b$ they committed to. 

We now give a formal definition.

\begin{definition}[Quantum bit commitment]\label{def:commitment_scheme}
Let $\lambda \in \mathbb{N}$ denote the security parameter. A quantum bit commitment scheme is a uniform family of unitary quantum circuits $\{ Q_{\lambda}^b \}_{\lambda \in \mathbb{N},b \in \{0,1\}}$ where for each $\lambda$, the circuits $Q_{\lambda}^0,Q_{\lambda}^1$ act on $n=\poly(\lambda)$ qubits and output two registers $\mathsf{C},\mathsf{R}$. The scheme consists of two separate phases:
\begin{enumerate}
    \item (\textbf{Commit phase:}) to commit to a bit $b \in \{0,1\}$, the sender prepares the state $$\ket{Q_{\lambda}^b}_{\mathsf{RC}} = Q_{\lambda}^b \ket{0^{n}}$$ and then sends the ``commitment register'' $\reg{C}$ to the receiver. 
    
    \item (\textbf{Reveal phase:}) the sender announces the bit $b$ and sends the ``reveal register'' $\mathsf{R}$ to the receiver. The receiver then accepts if performing the inverse unitary $Q_{\lambda}^{b,\dagger}$ on registers $\mathsf{C},\mathsf{R}$ and measuring in the computational basis yields the state $\ket{0^n}$.
\end{enumerate}
For security, we require that the following two properties hold:
\begin{itemize}
    \item (\textbf{Stat. Hiding:}) For every quantum algorithm $\algo A_\lambda$ with single-bit output, it holds that
    $$
\Big | \Pr \Big [ \algo A_\lambda(\rho_{\lambda}^0) = 1 \Big ]- \Pr \Big [ \algo A_\lambda(\rho_{\lambda}^1) = 1\Big ] \Big | \leq \negl(\lambda) \, ,
$$
where $\rho_{\lambda,b}$ denotes the reduced density matrix of $\ket{Q_{\lambda}^b}$ on register $\mathsf{C}$. 

\item (\textbf{Comp. Binding:}) For every efficient quantum algorithm $\algo A_\lambda$ acting on $\mathsf{R}$, it holds that
$$
 \mathrm{F}\left( \Big( \algo A_{\lambda} \otimes \id_{\mathsf{C}} \Big)(
 \proj{Q_{\lambda}^0}) , \proj{Q_{\lambda}^1} \right ) \leq \negl(\lambda).$$
\end{itemize}
\end{definition}

We now show the following theorem.

\begin{theorem}
Let $\lambda \in \mathbb{N}$ denote the security parameter. Let $k \in \mathbb{N}$ and $n \geq 8k$ be integers which are polynomial in $\lambda$.
Let
$\mathcal{D}_p^{\otimes n}$ denote the $n$-qubit depolarizing channel, for $p = O(1)$. Let $U \sim \Cliff_n$ be a random Clifford. Consider the pair of quantum circuits $(Q^0,Q^1)$ given by
\begin{align*}
\ket{Q^0}_{\mathsf{CR}} &= 
\sqrt{2^{-k}}\sum_{x} \sum_{E_a} \sqrt{\Pr_{E_a \sim \mathcal{D}_{p}^{\otimes n} }[E_a]} \,\Big(\ket{x} \otimes \ket{a}\Big)_{\mathsf{C}} \otimes \left(E_{a} \, U (\ket{0^{n-k}} \otimes \ket{x}) \otimes \ket{0}\right)_{\mathsf{R}}\\    
\ket{Q^1}_{\mathsf{CR}} &= \sqrt{2^{-k}}\sum_{x} \sum_{E_a} \sqrt{\Pr_{E_a \sim \mathcal{D}_{p}^{\otimes n} }[E_a]} \,\Big(\ket{x} \otimes \ket{a}\Big)_{\mathsf{C}} \otimes \left(\ket{0^{n-k}} \otimes \ket{x} \otimes \ket{a}\right)_{\mathsf{R}}.     
\end{align*}
Then, assuming the hardness of the $\mathsf{LSN}_{n,k,\mathcal{D}_p^{\otimes n}}$ problem, the pair $(Q^0,Q^1)$ is a statistically hiding and computationally binding quantum bit commitment scheme.
\end{theorem}
\begin{proof} From \Cref{lem:fidelity-lemma}, we know that for $n \geq 8k$ and $p=O(1)$, e.g., $p=0.04$, with overwhelming probability over the choice of the random Clifford $U$, it holds that
$$
\delta_{\mathsf{TD}}(Q_{\mathsf{C}}^0,Q_{\mathsf{C}}^1) \leq 
\sqrt{1-
\mathrm{F}(Q_{\mathsf{C}}^0,Q_{\mathsf{C}}^1)} \leq 2 \cdot e^{-\frac{np}{48}} \leq \negl(\lambda).
$$
This implies that $(Q^0,Q^1)$ is a statistically hiding. The computational binding property follows immediately from the hardness of the $\mathsf{LSN}_{n,k,\mathcal{D}_p^{\otimes n}}$ problem. This is because a successful adversary against the computational binding property would allow us to solve $\mathsf{LSN}_{n,k,\mathcal{D}_p^{\otimes n}}$ in polynomial-time with success probability at least $1/\poly(n)$.

\end{proof}

\printbibliography
\end{document}